\documentclass{article}
\usepackage{graphicx}
\usepackage[colorlinks,citecolor=blue]{hyperref}
\usepackage{parskip,fullpage}
\usepackage{amsmath,amsfonts,amsthm,amssymb,physics}
\usepackage{cleveref}

\usepackage[
	lambda,
	operators,
	advantage,
	sets,
	adversary,
	landau,
	probability,
	notions,	
	logic,
	ff,
	mm,
	primitives,
	events,
	complexity,
	asymptotics,
	keys]{cryptocode}

\addtocontents{toc}{\protect\setcounter{tocdepth}{2}}

\renewcommand{\epsilon}{\varepsilon}
\newcommand{\F}{\mathbb{F}}
\newcommand{\N}{\mathbb{N}}
\newcommand{\Z}{\mathbb{Z}}
\newcommand{\R}{\mathbb{R}}
\newcommand{\from}{\gets}

\newcommand{\E}{\mathop{{}\mathbb{E}}}
\newcommand{\Var}{\mathop{{}\mathsf{Var}}}

\DeclareMathOperator{\calC}{\mathcal{C}}

\DeclareMathOperator{\calE}{\mathcal{E}}

\DeclareMathOperator{\calG}{\mathcal{G}}

\DeclareMathOperator{\calH}{\mathcal{H}}

\DeclareMathOperator{\calU}{\mathcal{U}}
\DeclareMathOperator{\calS}{\mathcal{S}}

\DeclareMathOperator{\calW}{\mathcal{W}}

\DeclareMathOperator{\Ber}{\mathrm{Ber}}
\DeclareMathOperator{\Hyp}{\mathrm{Hyp}}
\DeclareMathOperator{\wt}{\mathrm{wt}}

\DeclareMathOperator{\len}{\mathrm{len}}

\newcommand{\bias}{\mathsf{bias}}

\newcommand{\KeyGen}{\mathsf{KeyGen}}
\newcommand{\Encode}{\mathsf{Encode}}
\newcommand{\Decode}{\mathsf{Decode}}

\newcommand{\Enc}{\mathsf{Enc}}
\newcommand{\Dec}{\mathsf{Dec}}
\newcommand{\m}{{\sf m}}

\newcommand{\PRF}{\mathsf{PRF}}
\newcommand{\PRG}{\mathsf{PRG}}
\newcommand{\PRC}{\mathsf{PRC}}
\newcommand{\xLPN}{\mathsf{xLPN}}
\newcommand{\LPN}{\mathsf{LPN}}

\newcommand{\LDPC}{\mathsf{LDPC}}
\newcommand{\LDPCPRC}{\mathsf{LDPC\text-PRC}}

\newcommand{\PRCsharp}{\mathsf{PRC}^{sharp}}
\newcommand{\PRCsharppub}{\mathsf{PRC}^{sharp-pub}}

\newcommand{\ECC}{\mathsf{ECC}}
\newcommand{\PermBits}{\mathsf{PermBits}}

\newcommand{\Gideal}{\calG^{\mathsf{ideal}}}
\newcommand{\Greal}{\calG^{\mathsf{real}}}
\newcommand{\Grobustsk}{\calG^{\mathsf{robust-sk}}}
\newcommand{\Grobustpk}{\calG^{\mathsf{robust-pk}}}
\newcommand{\transcript}{\mathsf{transcript}}
\newcommand{\Gpub}{{\cal G}^{\mathsf{pub}}}

\newtheorem{theorem}{Theorem}
\newtheorem*{theorem*}{Theorem}
\newtheorem{corollary}[theorem]{Corollary}

\newtheorem{lemma}{Lemma}[section]
\newtheorem*{lemma*}{Lemma}

\newtheorem*{remark*}{Remark}

\newtheorem{definition}{Definition}
\newtheorem*{definition*}{Definition}
\newtheorem{construction}{Construction}
\newtheorem{assumption}{Assumption}

\newtheorem{fact}{Fact}[section]

\Crefname{fact}{Fact}{Facts}

\definecolor{darkgreen}{rgb}{0.0, 0.5, 0.0}

\title{Ideal Pseudorandom Codes}
\author{
    Omar Alrabiah\thanks{UC Berkeley. Email: \texttt{oalrabiah@berkeley.edu}}
    \and Prabhanjan Ananth\thanks{UCSB. Email: \texttt{prabhanjan@cs.ucsb.edu}}
    \and Miranda Christ\thanks{Columbia University. Email: \texttt{mchrist@cs.columbia.edu}}
    \and Yevgeniy Dodis\thanks{NYU. Email: \texttt{dodis@cs.nyu.edu}}
    \and Sam Gunn\thanks{UC Berkeley. Email: \texttt{gunn@berkeley.edu}}
}

\date{\today}

\begin{document}

\maketitle

\begin{abstract}
\noindent 
Pseudorandom codes are error-correcting codes with the property that no efficient adversary can distinguish encodings from uniformly random strings.
They were recently introduced by Christ and Gunn [CRYPTO 2024] for the purpose of watermarking the outputs of randomized algorithms, such as generative AI models.
Several constructions of pseudorandom codes have since been proposed, but none of them are robust to error channels that depend on previously seen codewords.
This stronger kind of robustness is referred to as \emph{adaptive robustness}, and it is important for meaningful applications to watermarking.

\vspace{2ex}
\noindent
In this work, we show the following.
\begin{itemize}
    \item {\bf Adaptive robustness}: We show that the pseudorandom codes of Christ and Gunn are adaptively robust, resolving a conjecture posed by Cohen, Hoover, and Schoenbach [S\&P 2025].
    Our proof involves several new ingredients, combining ideas from both cryptography and coding theory and taking hints from the analysis of Boolean functions.

    \item {\bf Ideal security}: We define an \emph{ideal pseudorandom code} as one which is indistinguishable from the ideal functionality, capturing both the pseudorandomness and robustness properties in one simple definition.
    We show that any adaptively robust pseudorandom code for single-bit messages can be bootstrapped to build an ideal pseudorandom code with linear information rate, under no additional assumptions.
    
    \item {\bf CCA security}: In the setting where the encoding key is made public, we define a CCA-secure pseudorandom code in analogy with CCA-secure encryption.
    We show that any adaptively robust public-key pseudorandom code for single-bit messages can be used to build a CCA-secure pseudorandom code with linear information rate, in the random oracle model.
\end{itemize}

\vspace{2ex}
\noindent
Together with the result of Christ and Gunn, it follows that there exist ideal pseudorandom codes assuming the $2^{O(\sqrt{n})}$-hardness of LPN. This extends to CCA security in the random oracle model.
These results immediately imply stronger robustness guarantees for generative AI watermarking schemes, such as the practical quality-preserving image watermarks of Gunn, Zhao, and Song (2024).
\end{abstract}

\thispagestyle{empty}
\newpage
{\small\tableofcontents}
\thispagestyle{empty}
\newpage

\setcounter{page}{1}
\section{Introduction}
\paragraph{Pseudorandom codes.}
Suppose that Alice wishes to send a message to Bob over some channel.
The fields of cryptography and coding theory each address a different concern that Alice might have about the channel:
\begin{itemize}
    \item if the channel is untrusted in the sense that there may be an eavesdropper, then we use encryption;
    \item if the channel is unreliable in the sense that information may be corrupted in transit, then we use error-correcting codes.
\end{itemize}
If the channel is both untrusted and unreliable, there is typically no issue with combining these techniques.
That is, if Alice wishes only to hide the \emph{content} of her message, then she can simply error correct an encryption of her message.
However, this is insufficient if Alice \emph{doesn't want the channel to even know that communication has occurred}.\footnote{This kind of covert communication is the classic problem of \emph{steganography}, but classic stateless steganographic techniques fail on unreliable channels.}

Concretely, imagine that Alice wishes for it to appear as if her transmissions are uniformly random.
What Alice needs is then \emph{robust, pseudorandom encryption} --- transmissions should appear random, and decryption should function even if they are corrupted.
Such an object is called a \emph{pseudorandom code} (PRC).

More formally, a (secret-key) PRC is a keyed error correction scheme consisting of algorithms for encoding and decoding.
To an adversary without knowledge of the key, an encoding oracle is computationally indistinguishable from an oracle that outputs a freshly random string on each query.
With the secret key, one can decode codewords even after they are subjected to an error channel.\footnote{We are interested in error channels that introduce a \emph{constant} rate of errors, and reserve the term ``pseudorandom code'' for schemes that handle such channels.
This is the more practically relevant as well as theoretically interesting setting; it is easy to handle any sub-constant error rate using just one-way functions \cite{GG24}.}

To build a PRC, the basic results of cryptography and coding theory fall short.
The central difficulty is that Alice's transmission must, on the one hand, appear \emph{completely structure-less} to the channel, and on the other, appear \emph{highly redundant} to Bob so that he can reliably recover the message even if the transmission is corrupted.
Therefore pseudorandom codes force us to push the boundaries of both cryptography and coding theory, a tension that is reflected in the fact that all known constructions of PRCs have quasipolynomial-time distinguishing attacks \cite{CG24,GM24,GG24}.

\paragraph{Watermarking.}
PRCs are not only natural cryptographic and coding-theoretic objects, but also powerful tools for constructing watermarks for generative AI.
See \Cref{subsection:relationship-wat} for an explanation of how this works; for now let us just say that at a high level, this is because one can typically view generative AI models as approximate reparameterizations of ``content'' (e.g. images or text) into a space where the content has a natural distribution (e.g. normal or uniform).
With this observation, \cite{CG24,GM24} use PRCs to construct watermarks for language models that are \emph{undetectable} in that the watermarked model is computationally indistinguishable from the original model (thereby ensuring that the watermark does not degrade generation quality), and \emph{robust} in that the watermark tolerates a constant rate of errors.
Prior text watermarks were either only undetectable but not robust \cite{CGZ24}, or robust but significantly altered the model's output distribution \cite{ZALW24,KTHL24}.
Similarly, the only known quality-preserving and robust watermarks for image generation models uses a PRC.
This approach was demonstrated in practice by \cite{GZS24}.

These schemes essentially embed a PRC codeword in the space where content has a natural distribution; pseudorandomness of the PRC implies undetectability of the watermark, and robustness of the PRC implies robustness of the watermark.
As we explain in \Cref{subsection:relationship-wat}, PRCs are in fact \emph{equivalent} to undetectable and robust watermarks for language models.

\paragraph{Oblivious vs adaptive robustness.}
While existing PRCs (and their corresponding watermarks) tolerate a high rate of errors, these errors must be made \emph{obliviously}, i.e., by a memoryless channel.
This falls short of the worst-case error model commonly considered in coding theory.
Furthermore, errors introduced by realistic adversaries are non-oblivious: A realistic adversary can see multiple watermarked responses, and sometimes even query a detection oracle.
It turns out that PRCs (and watermarks) that are highly robust to oblivious attacks may be easily removable by such adversaries, as demonstrated by an example in \cite{CHS24}.\footnote{Let $\calW$ be a watermarked model.
Consider another watermarked model $\calW'$ that, given a prompt $\pi$, outputs $\calW(\pi)$ if $\pi$ is unwatermarked and outputs an \emph{unwatermarked response} if $\pi$ is watermarked.
Then $\calW'$ has the same oblivious robustness as $\calW$, because an oblivious adversary cannot find a watermarked $\pi$; however, it is not at all robust to an adversary that makes just two queries to $\calW'$: The adversary can simply obtain a watermarked response $x$ from $\calW'$, then query $\calW'$ on $x$ to obtain an unwatermarked response.
The adversary doesn't even need to make any edits!}
In fact, this is not just a theoretical possibility.
Several practical attacks leverage access to encoding and decoding oracles to remove or forge watermarks \cite{DBLP:conf/icml/0001SV24,DBLP:journals/corr/abs-2402-16187}.

Therefore it is of both practical and theoretical importance to have PRCs that are robust to non-oblivious adversaries.
We say that a PRC is \emph{adaptively robust} if it can handle adversaries who are given an encoding oracle (or key).
We introduce \emph{ideal security} and \emph{CCA security} to handle cases where the adversary additionally queries a detection oracle.
Adaptive robustness was previously studied in \cite{CHS24}, but they showed only that the watermark of \cite{CGZ24} was adaptively robust to a very low (specifically, inverse security parameter) rate of errors.
\cite{CHS24} conjectured but did not prove that the PRC of \cite{CG24} is adaptively robust.

\paragraph{This work.}
We prove that a slight modification of the PRC of \cite{CG24} is adaptively robust to a \emph{constant rate of substitutions}, resolving the conjecture of \cite{CHS24}.
That is, for any $\delta < 1/4$ (or $\delta < 1/2$ if we wish only to detect but not decode messages), we show that no adversary can produce an error of relative weight at most $\delta$ that causes the decoder to fail --- even if the adversary is provided with the encoding key.
We call this property \emph{adaptive $\delta$-robustness}.
Our proof involves an interesting combination of techniques from cryptography and coding theory.

We then ask whether it is possible to achieve robustness to an adversary with access to both an encoding oracle \emph{and} a decoding oracle.
Rather than contribute to a growing list of separately defined properties, we take a principled approach and consider the \emph{ideal functionality} of a PRC.
This functionality dictates how the code should behave to any adversary, regardless of the adversary's goal (e.g., distinguishing codewords from random, or mauling codewords).
We say that a PRC satisfies \emph{ideal security} (or is an \emph{ideal PRC}) if the real PRC encoding and decoding oracles are indistinguishable from this ideal functionality.
This definition combines soundness, pseudorandomness, and adaptive robustness into one simple definition.

We show a generic and simple transformation from \emph{any} adaptively robust PRC capable of encoding a single-bit message to an ideal PRC with linear information rate.
Our transformation requires only a pseudorandom function, the existence of which is implied by that of a PRC.
This transformation, applied to the adaptively robust PRCs discussed above, yields a construction of an ideal PRC.

Finally, we turn to the public-key setting, where both pseudorandomness and robustness are defined with respect to adversaries that have access not only to an encoding \emph{oracle}, but an encoding \emph{key}.
While ideal security is the strongest possible definition in the secret-key setting, there is no obvious ``ideal'' security notion in the public-key setting.
Instead, we define a strengthening of adaptive robustness that we call \emph{CCA security} in analogy with CCA secure encryption.
This definition generalizes standard CCA security for pseudorandom encryption.
In the random oracle model, we show a generic transformation that builds a CCA-secure PRC from any adaptively robust public-key PRC.

We summarize our four definitions of PRC robustness in \Cref{table:definitions}.

\begin{table}[ht]
    \centering
    \begin{tabular}{|c|c|c|}
        \hline
        & \textbf{No decoder access} & \textbf{Decoder access} \\ \hline
        \textbf{Secret-key} & Secret-key adaptive robustness & Ideal security \\ \hline
        \textbf{Public-key} & Public-key adaptive robustness & CCA security \\ \hline
    \end{tabular}
\begin{center}   
\caption{The four definitions of robustness considered in this work. ``Secret-key'' means the adversary is given oracle access to the encoder and ``public-key'' means the adversary is given the encoding key. ``No decoder access'' and ``decoder access'' distinguish whether the adversary is additionally given oracle access to the decoder. We satisfy the ``no decoder access'' column in \Cref{thm:zero-bit,thm:one-bit}, ideal security in \Cref{thm:secret-key}, and CCA security in \Cref{thm:public-key}.}
    \label{table:definitions}
\end{center}
\end{table}

\subsection{Results}

In this subsection we highlight our main results, which can be summarized follows: (a) We prove unconditionally that certain PRCs based on LDPC codes are (maximally) adaptively robust in both the secret- and public-key settings; (b) In the secret-key setting, we present a generic transformation converting any adaptively robust PRC into an ideal PRC; (c) In the public-key setting, we present a generic transformation converting any adaptively robust PRC into a CCA-secure PRC in the random oracle model.

In presenting our results, we use the term ``$\delta$-robust'' to indicate that the attacker can adaptively corrupt up to a $\delta$ fraction of the codeword.
By ``zero-bit'' PRC we mean a PRC that is only capable of encrypting a fixed message; by ``single-bit'' we mean that the PRC can encrypt 0 or 1.

Our first two results show essentially optimal robustness of the zero-bit LDPC-based PRC from \cite{CG24}, and our single-bit version of it.
Any zero-bit PRC is at most $1/2$-adaptively robust, since such errors can completely randomize the codeword.
Any single-bit PRC is at most $1/4$-adaptively robust, since an adversary can always construct a string that is $1/4$-close to two codewords encoding different messages (see the remark at the end of \Cref{subsection:techo-adaptive-robustness}).

\begin{theorem}[Adaptively robust zero-bit PRC]\label{thm:zero-bit}
    For any $\varepsilon > 0$, the public-key zero-bit pseudorandom code from \cite{CG24} is adaptively $(1/2-\varepsilon)$-robust for appropriate choice of parameters.
\end{theorem}

\begin{theorem}[Adaptively robust single-bit PRC]\label{thm:one-bit}
    For any $\varepsilon > 0$, our public-key single-bit pseudorandom code is adaptively $(1/4-\varepsilon)$-robust for appropriate choice of parameters.
\end{theorem}

Our secret-key transformation is lightweight in that it makes no additional assumptions beyond the security of the underlying PRC,\footnote{It makes use of a pseudorandom function, whose existence is implied by that of a PRC.} and it preserves the quantitative level of robustness.

\begin{theorem}[Ideal PRC]\label{thm:secret-key}
    Suppose that there exists a secret-key single-bit pseudorandom code that is adaptively $\delta$-robust.
    Then there exists a $\delta$-robust ideal pseudorandom code with linear information rate.
\end{theorem}

The result is somewhat worse for the public-key case, incurring a constant-factor loss in robustness and requiring the random oracle model.

\begin{theorem}[CCA-secure PRC]\label{thm:public-key}
    Suppose that there exists a public-key single-bit pseudorandom code that is adaptively $\delta$-robust.
    Then in the random oracle model, there exists an $\Omega(\delta)$-robust CCA-secure pseudorandom code with linear information rate.
\end{theorem}

By applying Theorems~\ref{thm:secret-key}~and~\ref{thm:public-key} to the PRCs from \cite{CG24}, which are pseudorandom assuming the certain-subexponential hardness of LPN and adaptively $\delta$-robust for $\delta < 1/4$ by Theorem~\ref{thm:one-bit}, we have the following corollary.

\begin{corollary}\label{cor:final}
Assuming the $2^{O(\sqrt{\secpar})}$-hardness of learning parity with noise (LPN), there exist separate linear-rate PRCs satisfying
\begin{itemize}
    \item $\delta$-robust ideal security for any $\delta < 1/4$, and
    \item $\delta$-robust CCA security for some $\delta = \Omega(1)$ in the random oracle model.
\end{itemize}
\end{corollary}
While the above robustness is optimal for the ideal PRC result, we do not achieve the optimal robustness for the CCA-secure PRC result.

\subsection{Relationship to watermarking} \label{subsection:relationship-wat}
Christ and Gunn \cite{CG24} presented a general template for using pseudorandom codes to watermark sufficiently high-entropy outputs of randomized algorithms.
We recall that template here for reference.

Consider a randomized algorithm that, given an input $x$ and a random seed $r$, generates an output $y$.
Suppose that there exists a \emph{randomness recovery} algorithm that, given $y$, outputs an approximation of the random seed $r$ used to generate it.
Randomness recovery algorithms indeed exist for generative models used in practice, and their approximation accuracy increases with the amount of entropy in the output distribution.
The basic observation is that replacing $r$ with an output of $\PRC.\Encode(1)$ immediately yields a watermark.
Pseudorandomness ensures that this replacement does not perceptibly alter the quality of the model.
Error correction enables detection, with the content being considered watermarked if $\PRC.\Decode({\sf RandomnessRecovery}(y)) = 1$.
This watermarking approach works for \emph{any} randomized algorithm with a randomness recovery algorithm whose approximation error is within the PRC's tolerance.

\begin{remark*}
    It is not a coincidence that randomness recovery algorithms exist in almost every generative AI framework.
    In some cases, randomness recovery can be viewed as an essential part of training.
    It is also useful to have randomness recovery in many independent generative AI applications, because it can allow for content manipulation in a more natural ``latent space'' where features correspond to higher-order features of the content.
\end{remark*}

It is known that this framework yields an undetectable and robust watermark for large language models \cite{CG24,GM24}.
This framework has also been demonstrated experimentally by \cite{GZS24}, which builds a robust undetectable watermarking scheme for generative image models (in particular, latent diffusion models) using essentially the same PRC as we consider here.

Not only are pseudorandom codes useful for constructing robust undetectable watermarks, but they are also necessary in a formal sense.
Consider a model that can be prompted to output a uniformly random binary string (e.g., one can ask a language model to output a random sequence of ``apple'' and ``orange''). 
Suppose we can undetectably watermark this model, with robustness to a given error channel.
This implies that on each query, the model outputs a fresh pseudorandom string such that the detector outputs ${\sf True}$ even if this string is subjected to this error channel.
Soundness (i.e., the low false positive rate) of the watermark ensures that unrelated strings decode to ${\sf False}$.
Thus, the model with this fixed prompt is \emph{exactly} an encoder for a zero-bit PRC, where the PRC decoder is the watermark detector.

% In this work, as in \cite{CG24}, we require \emph{binary-alphabet} PRCs.
% While binary-alphabet PRCs are more challenging to construct than larger-alphabet PRCs, they are necessary for practical robust watermarks.
% For language model watermarks, a common goal is robustness to a constant rate of word substitutions.
% If our PRC tolerates a constant rate of symbol substitutions, we must embed $O(1)$ PRC symbols per word in the response.
% Suppose that our alphabet size were $\Omega(\secpar)$, and that we had some fixed mapping from words to PRC symbols.
% In order to influence the model to output a next word with the correct PRC symbol with even constant probability, this next word distribution must have significant probability mass over $\Omega(\secpar)$ words!
% However, typical language models have on the order of one bit of entropy per word.

\paragraph{PRC robustness and watermark robustness.}
Typical randomness recovery algorithms for generative AI are crucially imperfect, even on unedited generations.
If one wishes to have a \emph{robust} watermark, then randomness recovery will be still more inaccurate.
Therefore, it is essential for applications to watermarking that we use a PRC with strong robustness.
The watermark will be detectable as long as the composition of an adversary's modifications to the output, and any intrinsic error from randomness recovery, falls within the class of error channels tolerated by the PRC.

For example, a PRC with robustness to a constant rate of errors chosen by an adversary \emph{that observes only the given output} translates to a watermark with robustness to an adversary making a single query to the generation algorithm.
In practice, however, an adversary attempting to remove the watermark can make arbitrarily many generation queries, translating to a PRC error channel that observes many codewords.
Therefore we need our PRC to be adaptively robust if we wish to handle such watermark removal attacks.

Similarly, a watermark that is robust to an adversary with access to a watermark detection oracle necessitates a PRC that is robust to an adversary with a decoding oracle.
This robustness is satisfied by the ideal PRCs that we define here.
Similarly, an adversary given a watermark detection oracle \emph{and the watermark generation key} is handled by our CCA-secure PRCs.

\paragraph{On robustness to substitutions.}
In this work, we only consider PRCs with robustness to substitutions.
This is sufficient for some watermarking applications, like those where the PRC is embedded in a \emph{semantic} representation of the content, and general edits in the content itself translate to substitutions in the semantic space.
This is the case, for instance, in the image watermark of \cite{GZS24} because the models they consider map between images and a \emph{latent space} consisting of fixed-length vectors.
They embed a PRC codeword in this latent representation of the generated image, and their decoder maps the given image back to its latent representation before decoding the PRC.
Since these latent vectors are of fixed size, modifications to the image translate to changes to components, but not insertions or deletions.

For other applications --- especially watermarking language models --- it can be useful to have a PRC with robustness to more general forms of edits.
Random deletions were considered in \cite{CG24}, and oblivious deletions and insertions were considered in \cite{GM24} (over a polynomial-sized alphabet).
However, we note that both of these methods work \emph{by reduction to the binary substitution channel}.
Therefore improved results about substitution channels may translate to improved results for editing channels, although we do not investigate this here.

\paragraph{On the alphabet size.}
One can think of the watermark generation process as sampling each ``component'' of the output $y$ to be correlated with the corresponding bit of the random seed $r$.
The amount of entropy in a given component $y_i$ determines how much signal from the corresponding seed bit $r_i$ it contains.
In order to plant one symbol of a PRC codeword in each component with significant signal, we require $\Omega(\abs{\text{PRC alphabet}})$ entropy per component.

For language models, these ``components'' might be the tokens of the response.
Typical language models have on the order of one bit of entropy per word, necessitating a PRC with a constant-sized (ideally, binary) alphabet.\footnote{While it is possible to increase the entropy per component by taking each component to be a sequence of $k$ words, this harms robustness as changing one in every $k$ words now changes every symbol of the underlying codeword.}
Other watermarking applications will also suffer from the use of a larger alphabet.
Therefore, in this work we only consider \emph{binary-alphabet} PRCs.

\paragraph{Related work.}
We have referred to most of the relevant existing works, particularly those about pseudorandom codes, throughout this introduction.
For a more detailed discussion on related work on watermarking, see \Cref{section:related-work}.

\section{Technical overview}

\subsection{Organization} \label{subsection:techo-organization}
Adaptive robustness, ideal security, and CCA security define robustness for a PRC under various adversarial models, as shown in \Cref{table:definitions}.
See \Cref{fig:flow} for a visual representation of the structure of the paper.

We briefly recall the definition of a pseudorandom code in \Cref{subsection:techo-prcs} (and \Cref{subsection:prelims-prcs}).

In \Cref{subsection:techo-adaptive-robustness} (and \Cref{section:adaptive-robustness}), we show that a particular PRC construction based on LDPC codes is adaptively robust in the public-key setting.
This immediately implies that the same construction is adaptively robust in the secret-key setting.
This is the only section that is particular to any specific PRC construction; all of our results in later sections are are generic in the underlying PRC.

In \Cref{subsection:techo-ideal-prcs} (and \Cref{section:ideal-security}), we show that in the secret-key setting, any single-bit adaptively $\delta$-robust PRC can be converted to a $\delta$-robust ideal PRC with a linear information rate.

\Cref{subsection:techo-cca} (and \Cref{section:cca-security}) mirrors the preceding section, but in the public-key setting.
As is the case with standard encryption, in the public-key setting there is no longer a clear notion of ``ideal security,'' but we nonetheless present a definition we call ``CCA security'' in analogy with CCA-secure encryption.
We show that in the random oracle model, any public-key adaptively $\delta$-robust PRC can be converted to an $\Omega(\delta)$-robust CCA-secure PRC with a linear information rate.

\subsection{Pseudorandom codes} \label{subsection:techo-prcs}
We recall the definition of a public-key PRC.

\begin{definition*}
    Let $\Sigma$ be an alphabet and $\calE : \Sigma^* \to \Sigma^*$ be an error channel. A \emph{public-key PRC} with oblivious robustness to $\calE$ is described by efficient randomized algorithms $\Encode_{\pk} : \Sigma^k \to \Sigma^n$ and $\Decode_{\sk} : \Sigma^* \to \Sigma^k \cup \{\bot\}$, parameterized by keys $\sk, \pk$, satisfying the following criteria for every security parameter $\secpar$:
    \begin{itemize}
        \item (Oblivious robustness) For any message $\m \in \Sigma^k$,
        \[
            \Pr_{\sk,\pk}[\Decode_{\sk}(\calE(x)) = \m : x \from \Encode_{\pk}(\m)] \geq 1-\negl.
        \]
        \item (Soundness) For any fixed $c \in \Sigma^*$,
        \[
            \Pr_{\sk}[\Decode_{\sk}(c) = \bot] \geq 1 - \negl.
        \]
        \item (Pseudorandomness) For any polynomial-time adversary $\adv$,
        \[
            \abs{\Pr_{\pk}[\adv^{\Encode_{\pk}}(\secparam, \pk) = 1] - \Pr_{\pk,\calU}[\adv^{\calU}(\secparam, \pk) = 1]} \leq \negl,
        \]
        where $\adv^{\calU}$ means that the adversary has access to an oracle that, on any (even previously queried) input, responds with a freshly drawn uniform value in $\Sigma^n$.
    \end{itemize}
\end{definition*}

If the scheme can only encode a singular message (i.e. $k = 0$), then we call it a \emph{zero-bit} PRC.

Observe that the robustness condition in this definition does not allow the error channel $\calE$ to use information about the PRC keys, or even previously seen codewords, to choose the error.
Prior results on PRCs only achieve this kind of robustness, which we refer to as \emph{oblivious robustness}.

\subsection{Adaptively robust public-key PRCs based on LDPC codes} \label{subsection:techo-adaptive-robustness}
In this subsection we will outline our proof that the zero-bit LDPC-based PRCs of \cite{CG24} are adaptively robust in the public-key setting.
We then show how to build a single-bit adaptively robust PRC, using a different construction from theirs.
These results immediately imply the corresponding results in the secret-key setting.

\paragraph{The zero-bit LDPC-based PRC.}
Let us begin by recalling the zero-bit LDPC-based PRC construction.
It will be useful to define the set of all $t$-sparse vectors in $\F_2^n$,
\[
    \calS_{t,n} = \{w \in \F_2^n : \wt(w) = t\}.
\]
The secret key consists of a collection of $r = n^{\Omega(1)}$ random parity checks $w_1, \dots, w_r \from \calS_{t,n}$, for some $t = \Theta(\log n)$.
Arrange these parity checks into a matrix $\sk = H \in \F_2^{r \times n}$, which will serve as the secret detection key.
The public encoding key consists of a random matrix $\pk = G \in \F_2^{n \times d}$ such that $H G = 0$, where $d = \Theta(\log^2 n)$.
Let $\calC$ be the image of $G$.

Since this is a zero-bit PRC, we only need to describe a procedure for encoding 1.
To encode 1, the encoder outputs $c \oplus e^*$, where $c \from \calC$ and $e^* \from \calS_{\eta n, n}$ for some small constant $\eta > 0$.\footnote{In \cite{CG24}, the error was chosen to be i.i.d Bernoulli. For technical reasons, in this work we use random errors of fixed weight instead. In this overview, we also use a star to distinguish the encoding noise $e^*$ from the adversarial perturbation $e$.}
To decode a string $x$, the decoder computes $\wt(H x)$, where $\wt$ is the Hamming weight.
If $\wt(H x)$ is significantly less than $r/2$, then the decoder outputs 1; otherwise it outputs $\bot$.

For an appropriate choice of parameters, this scheme is robust to a constant rate of errors \emph{if the errors are chosen independently of $H$}.
In that case, each parity check has an $\Omega(1)^t$ bias towards 0, which is significant if $t$ is small enough.
And since there are $r = n^{\Omega(1)}$ independent parity checks in $H$, a Chernoff bound implies that detection fails with negligible probability.
This was essentially the argument used by \cite{CG24} to show that their zero-bit LDPC-based PRCs are robust to a constant rate of errors that are oblivious to the PRC keys $H, G$.

Of course, this argument breaks down when the errors are allowed to depend on $H$: Indeed, there is a simple attack that uses $H$ to select $o(n)$ errors that fool the detector.\footnote{The algorithm uses Gaussian elimination to find an assignment of the first $k = o(n)$ coordinates such that the first $k$ parity checks in $H$ are unsatisfied, if the remaining $n-k$ coordinates are all 0.}
Instead, we will show that this scheme is robust against an adversary that is given $G$, but not $H$.
Following \cite{CHS24}, we refer to this kind of robustness as \emph{adaptive robustness}.

\begin{definition*}[\Cref{definition:adaptive-robustness-pk}, public-key adaptive $\delta$-robustness, informal]
    A public-key pseudorandom code $(\Encode, \Decode)$ is \emph{adaptively $\delta$-robust} if, for all efficient adversaries $\adv$,
    \[
        \Pr_{\sk, \pk}[\Decode(\sk, c \oplus e) \ne m \text{ and } \wt(e) \le \delta n \mid (m, r, e) \from \adv(\pk), c = \Encode(\pk, m; r)] \le \negl[n].
    \]
\end{definition*}

In this definition, we allow the adversary to select a message $m$ (which must be $m=1$ for a zero-bit PRC) and randomness $r$ that together define a codeword $c = \Encode(\pk, m; r)$.
The adversary wins if they find a perturbation $e$, which has weight at most $\delta n$, such that $c \oplus e$ does not decode to $m$.

The secret-key definition is similar, except that the adversary does not get to choose the randomness $r$ and instead interacts with $\Encode$ only via oracle access.

\begin{remark*}
    In this work we only consider robustness to substitution channels.
    See the paragraph ``PRC error model'' in \Cref{subsection:relationship-wat} for a discussion on this point.
\end{remark*}

\paragraph{Adaptive robustness of the zero-bit LDPC-based PRC.}
The main technical challenge that we overcome in this work is in proving that the zero-bit LDPC-based PRC is adaptively robust, i.e., that it is robust even when the adversary knows $G$.
More formally, we need to show that any low-weight error vector $e$ that the adversary comes up with must satisfy significantly more than $1/2$ of the parity checks in $H$.
Since we know that there \emph{exist} bad $o(n)$-weight errors that depend on $H$, our proof strategy must crucially use the fact that $e$ is computed without knowledge of $H$.

Our first key observation is that, from the perspective of the adversary, $H$ is a uniformly random collection of $r$ vectors chosen from all $t$-sparse parity checks consistent with $G$.
That is, the rows of $H$ can equivalently be sampled at random from\footnote{We define $\calC^\perp = \{w \in \F_2^n : w \cdot c = 0 \ \forall c \in \calC\}$.}
\[
    P_{\calC,t} = \calC^\perp \cap \calS_{t,n},
\]
where $\calC$ is the image of $G$, \emph{at the time of decoding}.
This is formalized in \Cref{lemma:adv-knowledge}.
Observe that if we choose the dimension of $\calC$ as $d = t \log n / 2$,
\[
    \abs{P_{\calC,t}} \approx \binom{n}{t} \cdot 2^{-d} \approx 2^{t \log n / 2},
\]
which is super-polynomial in $n$ since $t = \Theta(\log n)$.
Because there are so many parity checks in $P_{\calC,t}$, we expect that it should be difficult to find low-weight errors that fool $P_{\calC,t}$.

Indeed, our proof will proceed by showing that for \emph{any} low-weight error $e$ (depending arbitrarily on $G$), a random parity check from $P_{\calC,t}$ is satisfied by $e$ with probability significantly greater than $1/2$.
The adaptive robustness of our scheme will then follow by a Chernoff bound over the choice of $H$.

We now turn to our proof that every low-weight error $e$ satisfies significantly more than half of the parity checks in $P_{\calC,t}$.
First, for $S \subseteq \F_2^\ell$ and $z \in \F_2^\ell$, we define
\begin{align*}
    \bias(z) &= \E_{i \from [\ell]}[(-1)^{1\{z_i=1\}}] \text{\quad and} \\
    \bias(S, z) &= \E_{w \from S}[(-1)^{1\{w \cdot z=1\}}].
\end{align*}
The key ingredient is the following lemma, which serves as the technical backbone of \Cref{section:adaptive-robustness}.

\begin{lemma*}[\Cref{lemma:omar}, informal]
    For any code $\calC \subseteq \F_2^n$, there is a $\gamma_{\calC} \in \R$ such that the following holds. For every $x \in \F_2^n$,
    \[
        \bias(P_{\calC,t}, x) = \gamma_{\calC} \sum_{c \in \calC} \bias(\calS_{t,n}, (x \oplus c)).
    \]
    Furthermore,\footnote{Technically, this part is \Cref{fact:parity-concentration}.} if $\calC$ is a random linear code, then $\gamma_{\calC} = 1 \pm \negl[n]$ with probability $1-\negl[n]$.
\end{lemma*}

This lemma allows us to reason about the number of parity checks satisfied by $x$ under $P_{\calC,t}$, by reasoning instead about \emph{random parity checks} from $\calS_{t,n}$.
See \Cref{subsection:toolkit} for the proof of this lemma, which is short, elementary, and motivated by ideas from the analysis of Boolean functions.

If $x$ has low weight, then it turns out that the terms corresponding to non-zero $c$ typically have little contribution in the lemma.
This can be viewed as a consequence of the Johnson bound, which says that for a high-distance code there are not many codewords within any given Hamming ball of small radius.
Therefore, the lemma implies that\footnote{This approximation is not strictly true in general, but we prove it for the relevant regimes of $e$.}
\[
    \bias(P_{\calC,t}, e) \approx \bias(\calS_{t,n}, e).
\]
What we have described so far means that $P_{\calC,t} = \calC^\perp \cap \calS_{t,n}$ is enough to essentially capture the structure of all of $\calS_{t,n}$, for a random code $\calC$.
It is easy to see that $\bias(\calS_{t,n}, e) \approx \bias(e)^t$, so we can approximate this quantity as a function only of $\wt(e)$ (because $\bias(e)=1-2\wt(e)/n$).

The final step is to recall \Cref{lemma:adv-knowledge}, which says that we can sample the rows of $H$ as a random subset of $r$ parities from $P_{\calC,t}$ \emph{after} the adversary has decided on $e$.
By a Chernoff bound over the choice of $H$, it follows that $\bias(H e) \approx \bias(P_{\calC,t}, e)$. Together with the approximations $\bias(P_{\calC,t}, e) \approx \bias(\calS_{t,n}, e)$ and $\bias(\calS_{t,n}, e) \approx \bias(e)^t$, we finally have that
\begin{equation} \label{eq:techo-bias-approx}
    \bias(H e) \approx \bias(e)^t.
\end{equation}
With appropriate choice of parameters, this completes the proof of zero-bit adaptive robustness.

\paragraph{Adaptive robustness of single-bit LDPC-based PRCs.}
We now describe a single-bit PRC that is adaptively $\delta$-robust for any $\delta < 1/4$.
The reader might wonder why a zero-bit PRC cannot be immediately used as a single-bit PRC by re-interpreting $\bot$ as 0.
But with this re-interpretation, zero-bit robustness only requires that it is hard for an adversary to turn an encoding of 1 into an encoding of 0, while single-bit robustness \emph{also} requires that the adversary cannot turn an encoding of 0 into an encoding of 1.

Instead, our strategy is to generate two independent pairs of zero-bit adaptively $\delta$-robust PRC keys $(H_0, G_0)$ and $(H_1, G_1)$, and to define the encoding of $m \in \{0,1\}$ as a noisy sample from the image of $G_m$.\footnote{Actually we only allow the encoder to output \emph{non-zero} vectors. This is because the adversary in the public-key adaptive robustness game chooses the encoding randomness, and the zero vector would trivially violate robustness.}
The decoder outputs $m \in \{0,1\}$ if the given string decodes to 1 under the $H_m$ zero-bit decoder \emph{and} it decodes to $\bot$ under the $H_{1-m}$ zero-bit decoder.
If both zero-bit decoders output 1 or $\bot$, then our single-bit decoder cannot determine which bit is encoded, so it outputs $\bot$. 

Adaptive $\delta$-robustness of the zero-bit PRC implies that the adversary cannot find a $\delta n$-weight error that causes both zero-bit decoders to output $\bot$.
However, it says nothing about the possibility of causing both zero-bit decoders to output 1.
In other words, we still need to rule out the possibility that the adversary can sample a codeword under $G_0$ or $G_1$ and produce a nearby string that decodes to 1 under \emph{both} $H_0$ and $H_1$.

Let $\calC_m$ be the image of $G_m$ for $m \in \{0,1\}$.
The idea is to use the fact that $\calC_0$ and $\calC_1$ are well-separated with high probability: That is, the closest non-zero pair of codewords $c_0 \in \calC_0, c_1 \in \calC_1$ are approximately $n/2$-far with probability $1-\negl[n]$ over $G_0, G_1$.
Therefore, if the adversary produces $c \oplus e$ where $c \in \calC_m$ and $\wt(e) \le \delta n$, then $c \oplus e$ will be at least roughly $(1/2 - \delta) n$-far from any codeword in $\calC_{1-m}$.
For $\delta < 1/4$, this yields our separation: $c \oplus e$ is $\delta n < n/4$ far from $\calC_m$, and $(1/2-\delta) n > n/4$ far from $\calC_{1-m}$.

So the question is whether our parity check matrices $H_0, H_1$ will be able to observe this difference.
But fortunately, this was already addressed in \Cref{eq:techo-bias-approx}!
The only thing that remains is to account for parameters, making sure that the zero-bit decoder thresholds are set appropriately so as not to detect beyond the $\delta n$ radius desired.

\begin{remark*}
    While we construct zero-bit PRCs that are adaptively $\delta$-robust for any $\delta < 1/2$, it is not possible to construct single-bit PRCs for $\delta \ge 1/4$, as demonstrated by the following attack.
    The attacker draws random $x_0$ and $x_1$, encoding 0 and 1 respectively.
    By pseudorandomness, $x_0$ and $x_1$ must be equal on roughly half of their locations.
    Therefore, the adversary can craft $x'$ that differs from each of $x_0$ and $x_1$ on at most a $1/4$ fraction of locations.
    Since $x'$ cannot decode to both 0 and 1, robustness is violated for either $x_0$ or $x_1$.
\end{remark*}

\subsection{Ideal PRCs: the secret-key setting}
\label{subsection:techo-ideal-prcs}

So far, we've considered robustness against an adversary with access to only the encoder.
In this section, we give the adversary access to both encoding \emph{and} decoding oracles.\footnote{Since the adversary is only given to encoding and decoding oracles, this is the \emph{secret-key} setting. In the next section we will consider adversaries with access to an encoding key and a decoding oracle.}
An ideal PRC should retain both robustness and pseudorandomness in this setting.

Defining pseudorandomness requires some care in this case.
Of course, an adversary can distinguish a codeword from random simply by submitting the codeword as a decoding query.
Nonetheless, we find that robustness and pseudorandomness can be elegantly combined into a single definition in the secret-key setting, which we call \emph{ideal security}.
This definition uses the real/ideal world paradigm, common in cryptography.
Ideal security requires that no adversary can distinguish between the ideal world, where the challenger responds to encoding and decoding queries by comparing to previous responses; and the real world, where the challenger responds according to the PRC algorithms.

\paragraph{The ideal world.}
A PRC simply produces codewords that are pseudorandom, which can be decrypted even when subjected to errors.
Strings that are far from all observed codewords should decode to $\bot$.
Given this, an ideal $\delta$-robust PRC should satisfy the following requirements:
\begin{itemize}
    \item Responses to encoding queries should appear random.
    \item Any string that is within $\delta$ of an observed codeword should decode to its underlying message.
    \item Any string that is at least $\delta$-far from all observed codewords should decode to $\bot$.
\end{itemize}
These requirements \emph{fully specify} the behavior of the PRC in the ideal world.
Therefore, we define the ideal world as follows.
The challenger responds to each encoding query with a fresh uniformly random string, storing the queried message and response in memory.
The challenger responds to each decoding query by checking if the given string is close to any response in memory. If so, it returns the corresponding message; otherwise, it returns $\bot$.

\paragraph{Ideal security.}
A PRC satisfies ideal security if no adversary can distinguish between the real world and the ideal world.
Observe that ideal security constitutes a complete definition of a PRC: Whereas an adaptively robust PRC needs to satisfy separate definitions of soundness, pseudorandomness, and adaptive robustness, ideal security encompasses all of these properties at once.

\paragraph{Proving ideal security.}
We show that the existence of an adaptively $\delta$-robust single-bit PRC implies the existence of an ideal $\delta$-robust PRC with linear information rate.

We first show that the single-to-multi-bit PRC transformation from \cite{CG24} preserves adaptive robustness.
This transformation is simple.
Let $\PRC$ be any single-bit adaptively robust secret-key PRC, let $\PRG$ be a pseudorandom generator, and let $\ECC$ be an error-correcting code.
An encoding of $m$ takes the form
\begin{equation} \label{eq:single-to-multi-sk}
    \pi(\PRC.\Encode_{\sk}(r_1) || \ldots || \PRC.\Encode_{\sk}(r_\secpar), \PRG(r) \oplus \ECC(m)),
\end{equation}
where $r \gets \{0,1\}^\secpar$ is a fresh sample for each encoding and $\pi$ is a random permutation of the bits which is included as part of the key.
The decoder is straightforward.
The random permutation ensures that any errors introduced by the adversary are balanced across the blocks encoding the $r_i$'s and the error-corrected message.\footnote{This point is actually quite subtle. Because it is not possible to test whether an error is successful without the key, we cannot rely directly on pseudorandomness here. Instead, we use the fact that it \emph{is} possible to test whether the errors are balanced using only $\pi$, but not the underlying PRC key.}
Note also that if $\ECC$ has linear rate, then so does this resulting PRC (since $\secpar$ does not need to grow with the message length).

We now turn to showing that one can use any adaptively robust PRC with polynomial information rate, $\PRC$, to construct one satisfying ideal security.
The basic strategy is to build a scheme that satisfies ideal security in the case that the adversary makes only one decoding query.
Once this is established, complete ideal security will follow directly from a hybrid argument.

The main challenge is that, perhaps unintuitively, the PRC may sometimes be \emph{too} robust.
Ideal $\delta$-robustness requires a \emph{sharp} decoding threshold: any string that is $\delta$-far from all observed codewords must decode to $\bot$.
If our decoder accepts strings outside of this threshold, then it could accidentally leak information about the key!
Therefore, we need our decoder to be able to determine \emph{precisely} if a given string is within distance $\delta n$ of a previously seen codeword.

Suppose an adversary takes a codeword $c$ and adds some (possibly greater than $\delta n$-weight) error to obtain $c'$.
Our key idea is that if the decoder can somehow recover from $c'$ both the message and randomness used to generate $c$, then it can recover $c$ and exactly compute its distance from $c'$.
We define $\PRCsharp$ to allow exactly this, by drawing a random $r \from \{0,1\}^\secpar$ and letting the encoding of a message $\m$ under $\PRCsharp$ be\footnote{Note that this is essentially the Fujisaki-Okamoto transformation \cite{FO99}, but in a secret-key setting where a PRF suffices instead of a random oracle.}
\[
    c \gets \PRC.\Encode_{\sk}(r||\m; \PRF_{\sk}(r)),
\]
where $\PRF$ is a pseudorandom function.
In words, we encode a seed $r$ as part of the message, and use $\PRF_{\sk}(r)$ as the randomness in the encoding.
Pseudorandomness of $\PRF$ implies that codewords of $\PRCsharp$ are indistinguishable from those of $\PRC$.
Now suppose we are given $c'$ that is $\delta$-far from $c$ yet still decodes validly under $\PRC.\Decode_{\sk}$.
Our $\PRCsharp$ decoder will then recover $\m$ and $r$, recompute $c$, and observe that $c'$ is too far from $c$.
It then knows to output $\bot$.

We've now addressed the issue where $c'$ is too far from a previously seen $c$, but still decodes validly to the message encoded by $c$.
But there is still the possibility that the adversary produces a $c'$ that decodes to some message that it has never queried.
We address this issue by simply adding an authentication tag to the message.
That is, our final $\PRCsharp$ codewords take the form
\[
    c \gets \PRC.\Encode_{\sk}(r||\m||R_1; R_2),
\]
where $(R_1, R_2) = \PRF_{\sk}(r || \m)$.
Now, if a received string decodes to $r||\m||R_1$ under $\PRC.\Decode_{\sk}$, the $\PRCsharp$ decoder checks that $R_1 = \PRF_{\sk}(r || \m)$.
If so, it outputs $\m$; otherwise it outputs $\bot$.
Pseudorandomness of $\PRF$ implies that the adversary cannot produce an accepting $r, \m, R_1$ that were not included in a previous response.

Therefore, if the adversary produces $x$ that decodes to anything other than $\bot$, it must be the case that $x$ is $\delta$-far from some previously-seen codeword produced by the challenger.
This yields single-decoding-query ideal security, from which ideal security follows by a simple hybrid argument.

\begin{remark*}
    Multi-bit $\delta$-ideal security is only possible for $\delta < 1/4$, by the same example in the remark at the end of \Cref{subsection:techo-adaptive-robustness}.
    However, we proved that adaptive robustness is possible up to $1/2$ for \emph{zero-bit} PRCs.
    It would be interesting to construct zero-bit PRCs that satisfy $\delta$-ideal security for $\delta \ge 1/4$, but it appears that new ideas would be needed.
\end{remark*}

\subsection{CCA security: the public-key setting} \label{subsection:techo-cca}
In this subsection we will outline our definition and construction of CCA-secure PRCs.

\paragraph{Single-bit to multi-bit PRC transformation.} 
We first present a simple transformation from an adaptively robust single-bit PRC to an adaptively robust PRC with linear rate, in the public-key setting.
In the public-key setting, the transformation described earlier in \Cref{eq:single-to-multi-sk} is no longer robust because it relies on the secrecy of the permutation $\pi$.

Our solution is to first encode $r$ in an error-correcting code, and encode each bit of the resulting codeword separately under the single-bit PRC.
We also make the block encoding $r$ and the block encoding the message equal lengths.
Now, an adversary introducing a $\delta$ rate of errors can introduce at most a $2\delta$ rate of errors to either the message block or the $r$ block.
If the errors on the $r$ block are concentrated on any given single-bit PRC, then the error-correcting code handles them; if they are spread out among the single-bit PRCs, then the PRC robustness handles them.
This transformation is quite lossy in terms of robustness, which is at most $1/32$ regardless of the underlying PRC and error-correcting codes.
The information rate is also at most half that of the error-correcting code.
It would be interesting to come up with a better transformation, but we do not pursue this here.
In any case, we still achieve $\Omega(1)$-robustness and a linear information rate.

\paragraph{CCA security: definition.}
Our CCA (chosen codeword attack) security definition subsumes both pseudorandomness and adaptive robustness.
It also generalizes CCA security for pseudorandom encryption.
Roughly speaking, the security definition states that pseudorandomness of public-key PRC should be secure against computationally bounded adversaries even when given the public key and access to the decoding oracle.
This definition will be modeled similar to the private-key setting, except that we need to additionally take into consideration that the adversary can create its own codewords.
We formalize the security definition by again considering a real and a ``random'' experiment.
Whereas we refer to the latter as the ``ideal'' world in the secret-key setting, in the public-key setting we use a different word because it is not clearly ``ideal.''
In particular, in the public-key setting our definition of the random world \emph{still makes use of the real PRC decoder}.

\paragraph{The random world.}
In the random game, the adversary receives the following: 
\begin{itemize} 
\item The public encoding key $\pk$, 
\item Access to the random encoding oracle: the random encoding oracle on input a message $m$, outputs a random string $c$ as a codeword. 
\item Access to the corresponding decoding oracle: upon receiving a codeword $c'$, this oracle first checks if $c'$ is close to any of the codewords returned by the random encoding oracle. If so, it returns $m$, where $c$ was returned as a response to the encoding query $m$. If $c'$ is not close to any of the outputs of the encoding oracle then it responds according to the actual decoder on input $c'$.
\end{itemize} 

\paragraph{The real world.}
In the real game, the adversary receives the public encoding key $\pk$ and has access to the real encoding and decoding oracles. 

\noindent CCA security simply says that any computationally bounded adversary cannot distinguish whether it is participating in the random or the real game. 

\paragraph{Achieving CCA security generically.} We show how to achieve CCA security generically starting with any adaptively robust PRC with polynomial information rate in the random oracle model.
Our transformation is essentially identical to the ideal PRC transformation, except that we use a random oracle instead of a PRF.
This makes it very similar to the Fujisaki-Okamaoto (FO) transformation~\cite{FO99}, although the analysis is slightly different.
Specifically, we crucially rely upon the adaptive robustness property in the security proof unlike the FO transformation.

Suppose $\PRC$ is any adaptively robust pseudorandom code supporting multi-bit messages.
We construct a CCA-secure pseudorandom code $\PRC^{{\sf CCA}}$ as follows: 
\begin{itemize}

    \item The key generation algorithm of $\PRC^{{\sf CCA}}$ simply runs the key generation of the $\PRC$ to generate $(\pk,\sk)$. It also samples $F$ from a hash function family; $F$ will be modeled as a random oracle in the security proof. The new public key of $\PRC^{{\sf CCA}}$ is set to be $(\pk,F)$ and the new secret key is set to be $\sk$. 
    
    \item The encode algorithm of $\PRC^{{\sf CCA}}$, on input a message $m$, does the following. It first samples a $\lambda$-bit string $r$ uniformly at random. It then computes $F(m,r)$ to obtain a string that can be broken down into two parts $(R_1,R_2)$, each of length at least $\lambda$. It then encodes the new message $(m,r,R_2)$ using the encode algorithm of $\PRC$ and using the randomness $R_1$. Denote the resulting codeword to be $c$. 
    
    \item The decode algorithm, on input the secret key $\sk$ and the codeword $c$, does the following: it first runs the $\PRC$ decoding algorithm on $c$ to obtain $(m,r,R'_2)$. It then computes $F(m,r)$ to obtain $(R_1,R_2)$. It outputs $m$ if $R_2=R'_2$. Otherwise, it outputs $\bot$. 
    
\end{itemize}
\noindent To prove the CCA security of the above scheme, we undertake the following three steps. 
\par In the {\bf first step}, we observe the queries made by the adversary. Specifically, we consider the following event. 

\begin{quote}
\noindent \underline{${\sf NeverQueried}$}: The adversary, during the decoding query phase, submits a codeword $c$ that decodes to $(m,r,R_2)$ such that $(m,r)$ has never been queried by the adversary to $F$. 
\end{quote} 

\noindent We argue that ${\sf NeverQueried}$ only happens with negligible probability. Indeed, if the adversary has never queried $F$ on $(m,r)$ then the probability that it predicts $R_2$ is negligible in $\lambda$. 
\par This suggests that there are only two types of decoding queries $c$ that the adversary can submit. Either
\begin{itemize} 

\item $c$ is close to one of the codewords returned by the challenger to the adversary during the encoding phase, or

\item $c$ decodes to $(m,r,R_2)$, where $(m,r)$ is a random oracle query made by the adversary. 

\end{itemize}
\noindent This suggests that the challenger does not need the decoding key to answer the decoding queries at all! It can answer just by (a) observing the adversarial queries to the $F$ and, (b) by keeping tracking of the codewords it returns during the encoding phase. 
\par This leads us to the {\bf second step}. In the second step, the challenger does not use the decoding key to answer the decoding queries. Instead it uses an alternate decoding procedure, referred to as ${\sf AltDecode}$. Upon receiving a codeword $c$ during the decoding query phase, ${\sf AltDecode}$ first performs the following checks: 
\begin{itemize} 
    \item It checks if $c$ is close to any of the codewords returned during the encoding phase, 
    \item It checks if $c$ is to close to any of the codewords created by the adversary using the queries to $F$. In more detail, for every adversarial query $(m,r)$, create a codeword with $(m,r,R_2)$ as the message and $R_1$ as the randomness, where $(R_1,R_2)$ is the output of $F$ on $(m,r)$. Check if $c$ is close to any of the created codewords. 
\end{itemize}
If any of these two checks pass,\footnote{It could be that there is more than one codeword that is close to $c$. In this case, pick one of them at random.} return the message encoded in the codeword close to $c$. Else, return $\bot$. 
\par To argue that ${\sf AltDecode}$ simulates the use of the real decoding algorithm during the decoding phase, we need to argue that for every decoding query, the output of ${\sf AltDecode}$ and the real decoder is the same. If not, then we claim that the adaptive robustness property is violated. This is because the adversary efficiently came up with two different, but close, codewords that open to two different messages. Thus, from the adaptive robustness property, we can conclude that the outputs of ${\sf AltDecode}$ and the real decoder are the same for every decoding query. 
\par In the {\bf third step}, we invoke the pseudorandomness guarantee of $\PRC$ to switch all the codewords generated during the encoding phase to be uniformly random strings. In order to do this switch, it was crucial that the challenger was using ${\sf AltDecode}$, which in turn does not use any secret key, in the decoding phase. 
\par All the three steps combined prove the CCA security of $\PRC^{{\sf CCA}}$.

\tikzstyle{startstop} = [rectangle, rounded corners, minimum width=3cm, minimum height=1.5cm,text centered, draw=black, fill=red!10]
\tikzstyle{private} = [rectangle, rounded corners, minimum width=3cm, minimum height=1.5cm,text centered, draw=black, fill=green!10]
\tikzstyle{public} = [rectangle, rounded corners, minimum width=3cm, minimum height=1.5cm,text centered, draw=black, fill=blue!10]

\tikzstyle{process} = [rectangle, minimum width=3cm, minimum height=1cm, text centered, draw=black, fill=orange!30]
\tikzstyle{decision} = [diamond, minimum width=3cm, minimum height=1cm, text centered, draw=black, fill=green!30]
\tikzstyle{arrow} = [thick,->,>=stealth]

\begin{figure}[!htb]
\centering
{\small 
\begin{tikzpicture}[node distance=3cm]
\node (init) {Public-key zero-bit non-adaptive PRC from \cite{CG24}};
\node (start) [startstop, below of=init, yshift=1cm] {Public-key zero-bit adaptive PRC};
\node (in1) [startstop, below of=start] {Public-key single-bit adaptive PRC};
\node (pro1) [private, below of=in1, xshift=-4cm] {Secret-key optimal-rate adaptive PRC};
\node (dec1) [private, below of=pro1] {Ideal PRC with optimal rate};
\node (pro2) [public, below of=in1, xshift=4cm] {Public-key linear-rate adaptive PRC};
\node (dec2) [public, below of=pro2] {CCA-secure PRC with linear rate};
\draw [arrow] (init) -- node[right=2mm] {\Cref{subsection:zero-bit}} (start);
\draw [arrow] (start) -- node[right=2mm] {\Cref{subsection:single-bit}} (in1);
\draw [arrow] (in1) -- node[left=2mm] {\Cref{subsection:rate-boost-sk}} (pro1);
\draw [arrow] (pro1) -- node[left=2mm] {\Cref{subsection:ideal-security-proof}} (dec1);
\draw [arrow] (in1) -- node[right=2mm] {\Cref{subsection:rate-boost-pk}} (pro2);
\draw [arrow] (pro2) -- node[right=2mm] {\Cref{subsection:cca-scheme}} (dec2);
\end{tikzpicture}
}
\caption{Organization of the paper. The red boxes are particular to the LDPC-based public-key PRCs related to \cite{CG24}, the green boxes are generic in the secret-key setting, and the purple boxes are generic in the public-key setting.} \label{fig:flow}
\end{figure}
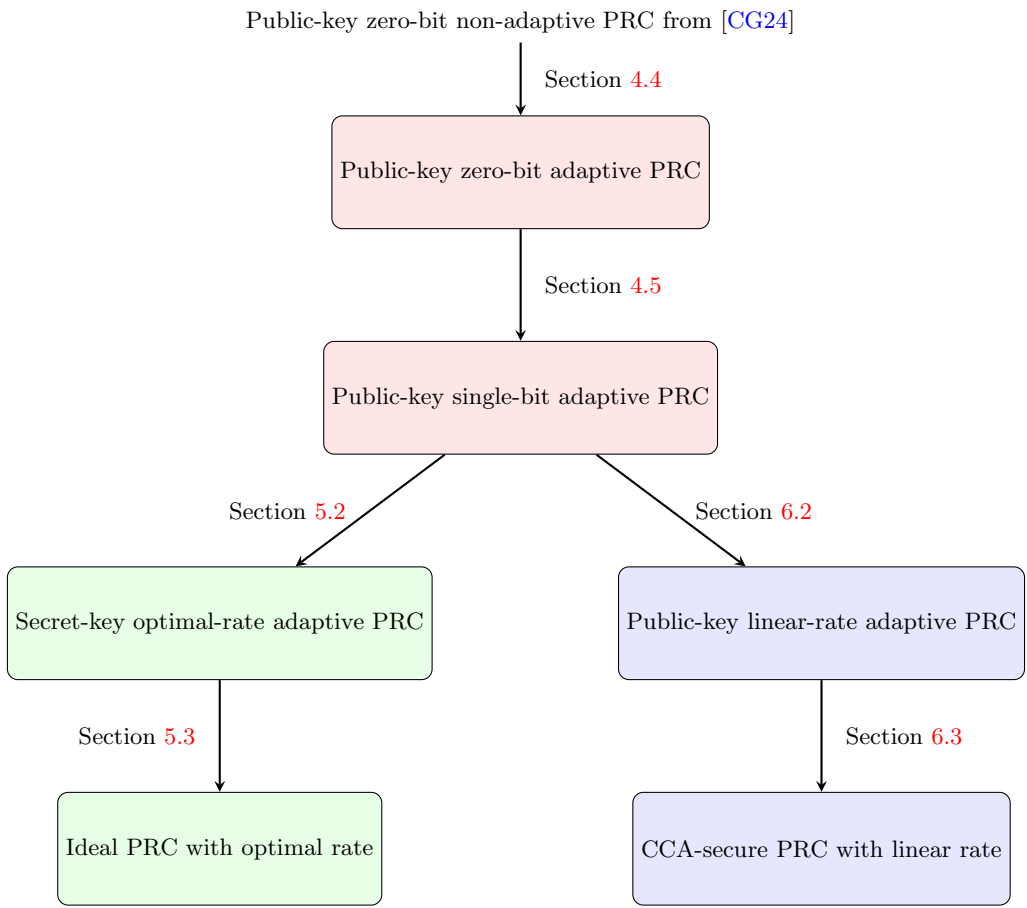

\section{Preliminaries} \label{section:prelims}

For a randomized algorithm $A(\cdot)$, we write $A(x;r)$ to denote the output of $A$ on input $x$ and randomness $r$.

For a set $X$, we define $X^* = \{(x_1, \dots, x_k) \mid x_1, \dots, x_k \in X \wedge k \in \Z_{\ge 0}\}$ to be the set of all strings over the alphabet $X$.
For a binary string $s \in X^*$, we let $s_i$ denote the $i^{\text{th}}$ symbol of $s$ and $\len s$ denote the length of $s$.

We write $[n] = \{1, \dots, n\}$.
Let $x \in \{0,1\}^n$ and let $\pi$ be a permutation over $[n]$.
We let $\PermBits(x, \pi)$ denote the function that outputs $x_{\pi(1)}|| \cdots || x_{\pi(n)} \in \{0,1\}^n$.

\begin{lemma}[Hypergeometric tail bounds \cite{hoeffding1994probability}] \label{theorem:hyper}
Let $X \sim \Hyp(N,K,n)$ and $p = K/N$. Then for any $0 < t < K/N$,
\begin{align*}
    \Pr \left[X \leq (p-t)n \right] &\leq e^{-2t^2 n}, \text{ and}\\
    \Pr \left[X \geq (p+t)n \right] &\leq e^{-2t^2 n}.
\end{align*}
\end{lemma}

\begin{lemma}[Chernoff bounds]
    Let $X_1, \ldots, X_n \in [0,1]$ be independent random variables. Let $\mu = \E\left[\sum_{i=1}^n X_i \right]$. Then for any $\delta \in (0,1)$:
    \begin{align*}
        \Pr \left[\sum_{i=1}^n X_i \geq (1+\delta)\mu \right] &\leq \exp \left(-\frac{\mu \delta^2}{3} \right) \text{ and}\\
        \Pr \left[\sum_{i=1}^n X_i \leq (1-\delta)\mu \right] &\leq \exp \left(-\frac{\mu \delta^2}{2} \right).
    \end{align*}
\end{lemma}

\paragraph{Coding theory notation.}
Let $n$ be the dimension, let $t$ be an even integer, and let $d = \Theta(\log \binom{n}{t})$.
Let $W_t \in \F_2^{\binom{n}{t} \times n}$ be the matrix whose rows are all weight-$t$ parity checks.
For a linear code $\calC \subseteq \F_2^n$, let $P_{\calC, t}$ be the matrix whose rows are all weight-$t$ parity checks satisfied by $\calC$.
Let $N_{\calC, t}$ be the number of rows in $P_{\calC, t}$.
For $z \in \{0,1\}^m$, we define $\bias(z) = \frac{1}{m} \sum_{i=1}^m (-1)^{z_i}$.
Note that $\bias(z) = 1 - 2\wt(z)/n$ and $\wt(z) = (1/2 - \bias(z)/2) \cdot n$.

\paragraph{Cryptography preliminaries.}
A pseudorandom function is a function that behaves indistinguishably from a random function, from the perspective of any computationally-bounded adversary that makes only black-box queries to the function.
Pseudorandom functions are equivalent to one-way functions \cite{GGM86}, the minimal object of classical cryptography.

\paragraph{Pseudorandom function (PRF).} Let $\mathcal{F} = \{F_{\sk} : \{0,1\}^{\ell_1(\secpar)} \to \{0, 1\}^{\ell_2(\secpar)} \ | \ \sk \in \{0,1\}^\secpar\}$ be a family of functions. 
$\mathcal{F}$ is a PRF if $F_\sk$ is efficiently computable and for all polynomial-time distinguishers $D$,
\[\left| \Pr_{\sk \gets \{0,1\}^\secpar}\left[D^{F_\sk(\cdot)}(1^\secpar) = 1\right] - \Pr_{f} \left[D^{f(\cdot)}(1^\secpar) = 1 \right]\right| \leq \negl.\]
where $f$ denotes a random function from $\{0,1\}^{\ell_1(\secpar)}$ to $\{0,1\}^{\ell_2(\secpar)}$.

\subsection{Pseudorandom codes} \label{subsection:prelims-prcs}
We recall the definition of a public-key pseudorandom code (PRC) with oblivious robustness from \cite{CG24}.
We present only the public-key definition; the secret-key version is identical except that the public key is included in the secret key, and the adversary in the pseudorandomness condition is only given $\secparam$ as input.
Secret-key PRCs are defined fully in \cite{CG24}.

\begin{definition}[Public-key PRC] \label{def:pkPRC}
    Let $\Sigma$ be a fixed alphabet. A \emph{public-key pseudorandom error-correcting code} (abbreviated as public-key PRC) with (oblivious) robustness to a channel $\calE : \Sigma^* \to \Sigma^*$ is a triple of polynomial-time randomized algorithms $(\KeyGen, \Encode, \Decode)$ satisfying
    \begin{itemize}
        \item (Syntax) There exist functions $\ell_\dec, \ell_\enc, n, k : \N \to \N$ such that for all $\secpar \in \N$, $\KeyGen(\secparam) \in \{0,1\}^{\ell_\dec(\secpar)} \times \{0,1\}^{\ell_\enc(\secpar)}$, $\Encode(\secparam, \pk, \m) \in \Sigma^{n(\secpar)}$ takes inputs $\pk \in \{0,1\}^{\ell_\enc(\secpar)}$, $\m \in \Sigma^{k(\secpar)}$, and $\Decode(\secparam, \sk, x) \in \Sigma^{k(\secpar)} \cup \{\bot\}$ takes inputs $\sk \in \{0,1\}^{\ell_\dec(\secpar)}$, $x \in \Sigma^*$.
        \item (Oblivious robustness) For any $\secpar \in \N$ and any message $\m \in \Sigma^{k(\secpar)}$,
        \begin{align*}
            \Pr_{(\sk,\pk) \from \KeyGen(\secparam)}[\Decode(\secparam,\sk,\calE(x)) = \m : x \from \Encode(\secparam,\pk,\m)] \geq 1-\negl.
        \end{align*}
        \item (Soundness) For any fixed $c \in \Sigma^*$,
        \[
            \Pr_{(\sk,\pk) \from \KeyGen(\secparam)}[\Decode(\secparam,\sk,c) = \bot] \geq 1 - \negl.
        \]
        \item (Pseudorandomness) For any polynomial-time adversary $\adv$,
        \[
            \abs{\Pr_{(\sk,\pk) \from \KeyGen(\secparam)}[\adv^{\Encode(\secparam,\pk, \cdot)}(\secparam,\pk) = 1] - \Pr_{\substack{(\sk,\pk) \from \KeyGen(\secparam) \\ \calU}}[\adv^{\calU}(\secparam,\pk) = 1]} \leq \negl,
        \]
        where $\adv^{\calU}$ means that the adversary has access to an oracle that, on any (even previously queried) input, responds with a freshly drawn uniform value in $\Sigma^{n(\secpar)}$.
    \end{itemize}
\end{definition}

We say that a PRC is ``zero-bit'' if it can encode only a singular message ($\Sigma = \{1\}$ and $k=0$), or ``single-bit'' if it can encode two messages ($\Sigma = \{0,1\}$ and $k=1$).

For completeness, we include the original oblivious robustness definition from \cite{CG24}. 
However, in this work we focus on stronger notions of robustness, which in particular imply oblivious robustness to substitution channels.

\section{Adaptively robust public-key PRCs based on LDPC codes}
\label{section:adaptive-robustness}

This section is dedicated to proving that variants of (public-key) PRCs from \cite{CG24} satisfy the adaptive robustness notion (without decoder access) we introduce in this work. Correspondingly, in Section~\ref{sec:def-adaptive} we give our new notion of adaptive robustness for both the public-key and the secret-key PRC variants (where the former is clearly stronger than the latter). In Section~\ref{sec:cg24} we review the {\em zero-bit} (public-key) PRC from \cite{CG24}, in Section~\ref{subsection:toolkit} we prove some technical tools that we will use, in Section~\ref{subsection:zero-bit} we prove the (public-key, and hence, also secret-key) adaptive security of this construction. And, finally, in Section~\ref{subsection:single-bit} we show how to extend the resulting zero-bit  construction to a single-bit adaptively-secure (still public-key) PRC.

\subsection{Definitions}\label{sec:def-adaptive}
We define adaptive robustness in the secret-key setting via the following security game.

\noindent \underline{$\Grobustsk_{\adv, \PRC, \delta}(1^\secpar)$:}
\begin{enumerate}
    \item The challenger sets $\transcript = \emptyset$ and samples $\sk \gets \PRC.\KeyGen(1^\secpar)$.
    \item The adversary is allowed to make encoding queries. For each encoding query $m$, the challenger responds with $c \gets \PRC.\Encode(\sk, m)$ and sets $\transcript = \transcript \cup \{(m, c)\}$.
    \item The adversary sends the challenger $x$.
    \item The challenger computes $m' = \PRC.\Decode(\sk; x)$ (which could be $\bot$).
    \item If there exists $(m, c) \in \transcript$ such that $\wt(x \oplus c) \le \delta n$ and $m \ne m'$, then the adversary wins; otherwise the adversary loses.
\end{enumerate}

\begin{definition}[Secret-key adaptive robustness] \label{definition:adaptive-robustness-sk}
    We say that a secret-key pseudorandom code $\PRC$ is \emph{adaptively $\delta$-robust} if, for any efficient adversary $\adv$,
    \[
        \Pr[\adv \text{ wins } \Grobustsk_{\adv, \PRC, \delta}(1^\secpar)] \le \negl.
    \]
\end{definition}

The public-key setting is similar, except that the adversary gets the encoding key itself rather than merely oracle access to the encoder.
Since the adversary can therefore encode messages on their own, we require the adversary to submit a message and randomness that witness the failure of the decoder. 

\noindent \underline{$\Grobustpk_{\adv, \PRC, \delta}(1^\secpar)$:}
\begin{enumerate}
    \item The challenger samples $(\pk, \sk) \gets \PRC.\KeyGen(1^\secpar)$ and sends $\pk$ to the adversary.
    \item The adversary sends $(m, r, x)$ to the challenger.
    \item The challenger computes $c = \PRC.\Encode(\pk, m; r)$ and $m' = \PRC.\Decode(\sk, x)$. If $\wt(x \oplus c) \le \delta n$ and $m' \ne m$, then the adversary wins; otherwise the adversary loses.
\end{enumerate}

\begin{definition}[Public-key adaptive robustness] \label{definition:adaptive-robustness-pk}
    We say that a public-key pseudorandom code $\PRC$ is \emph{adaptively $\delta$-robust} if, for any efficient adversary $\adv$,
    \[
        \Pr[\adv \text{ wins } \Grobustpk_{\adv, \PRC, \delta}(1^\secpar)] \le \negl.
    \]
\end{definition}

Note that \Cref{definition:adaptive-robustness-pk} is a stronger definition than \Cref{definition:adaptive-robustness-sk} --- any scheme satisfying \Cref{definition:adaptive-robustness-pk} automatically satisfies \Cref{definition:adaptive-robustness-sk}.
We therefore devote this section to proving public-key adaptive robustness of the PRC from \cite{CG24}, which immediately implies the secret-key adaptive robustness of the same PRC.

At this point, a few remarks are in order about our choice of definitions.

\begin{remark*}
    The reader may wonder why we don't just have the adversary produce two nearby strings $x, x'$ that decode to different values.
    This would be an impossibly strong definition: The adversary could choose $x$ with slightly fewer errors than the scheme tolerates, and $x'$ with slightly more.
    A successful adversary must come up with a \emph{valid} codeword, together with a low-weight modification of it that decodes incorrectly.
\end{remark*}

\begin{remark*}
It is possible to define public-key robustness along the same lines as~\Cref{definition:adaptive-robustness-sk} (secret-key adaptive robustness).
That is, we could consider the following game: the adversary has adaptive access to the encoding oracle and later on, is expected to come up with a string $c$ that is close to one of the codewords, say $c'$, returned by the encoding oracle.
It wins if $c'$ is an encoding of $m$ and $c$ does not decode to $m$.
There are a couple of reasons behind our choice of the adaptive robustness definition in the public-key setting (\Cref{definition:adaptive-robustness-pk}). First,~\Cref{definition:adaptive-robustness-pk} is stronger than the public-key analogue of~\Cref{definition:adaptive-robustness-sk}, and we will see that our scheme satisfies this stronger notion anyways.
Second, the definition is more compact and easier to work with. In particular, we crucially invoke this stronger~\Cref{definition:adaptive-robustness-pk} in our later proofs about CCA security (see the proof of~\Cref{lem:hybrid2:hybrid3}).
\end{remark*}

\begin{remark*}
One can also similarly define a stronger version of our symmetric-key \Cref{definition:adaptive-robustness-sk}, where the attacker can also control the randomness for the encoding oracle (and not only the message).
Of course, our scheme will satisfy this stronger definition too (as it satisfies the public-key analog of this strengthening given in ~\Cref{definition:adaptive-robustness-pk}).
However, we did not choose to follow this route for several reasons.
First, unlike the public-key case, this definition is not significantly more compact or intuitive than our definition.
Second, we do not have any real-world motivation for this notion, and unlike the public-key case, we do not require the stronger notion for any proofs later on in the paper.
Third, unlike the public-key setting, there is a simple generic transformation from our \Cref{definition:adaptive-robustness-sk} to the stronger variant: instead of sampling randomness $r$ for $\PRC.\Encode$ directly, we sample auxiliary randomness $s$, and set $r = \PRF_{\sk}(m,s)$, where $\PRF_{\sk}$ is a pseudorandom function whose key is part of the overall secret key.
\end{remark*}

\subsection{The scheme}\label{sec:cg24}
Here we recall the LDPC-based zero-bit PRC construction from \cite{CG24}.
For technical reasons, we modify the definition slightly to use fixed-weight error instead of Bernoulli.

Let
\[
    \calS_{t,n} = \{s \in \F_2^n : \wt(s) = t\}
\]
be the set of all $t$-sparse vectors in $\F_2^n$, and
\[
    \calS_{t,r,n} = \{H \in \F_2^{r \times n} : \wt(H_{i,:}) = t\ \forall i \in [r]\}
\]
be the set of all $t$-row-sparse matrices in $\F_2^{r \times n}$.

Our zero-bit pseudorandom LDPC codes are parameterized by a public generator matrix $G \in \F_2^{n \times d}$ and a secret parity-check matrix $H \in \F_2^{r \times n}$. The sampling process for these matrices is described in \Cref{def:random-ldpc}.

\begin{definition}[Random LDPC code, {$\LDPC[n,d,t,r]$}] \label{def:random-ldpc}
    \sloppy
    For $n, d, t, r \in \N$, define the distribution $\LDPC[n,d,t,r]$ over $\F_2^{r \times n} \times \F_2^{n \times d}$ as follows:
    \begin{enumerate}
        \item[] $\LDPC[n,d,t,r]$:
        \item Sample $H \from \calS_{t,r,n}$, i.e. $H \in \F_2^{r \times n}$ is chosen to have i.i.d random $t$-sparse rows.
        \item Sample $G \from (\ker H)^d$, i.e. $G \in \F_2^{n \times d}$ is a random matrix subject to $H G = 0$.
        \item Output $(H,G)$.
    \end{enumerate}
    An $(n,d,t,r)$ \emph{random LDPC code} is a pair of matrices $(H, G) \from \LDPC[n,d,t,r]$.
\end{definition}

We now define our LDPC-based zero-bit PRC. Recall that a zero-bit PRC is one whose message space is just $\{1\}$.
The following construction differs slightly from that in \cite{CG24}, in that the error distribution is uniform over $S_{\eta n, n}$ instead of $\Ber(n, \eta)$.

\begin{construction}[Zero-bit public-key pseudorandom LDPC code, {$\LDPCPRC_0[n,d,t,r,\eta,\zeta]$}] \label{const:zero-bit-ldpc-prc}
    Let $n, d, t, r : \mathbb{N} \to \mathbb{N}$ and $\eta, \zeta : \N \to [0,1/2)$ be efficiently-computable functions of the security parameter. We define $\LDPCPRC_0[n,d,t,r,\eta,\zeta]$ by the following algorithms, where we leave the dependence of $n,d,t,r,\eta,\zeta$ on $\secpar$ implicit:
    \begin{itemize}
        \item $\KeyGen(\secparam)$: Sample $(H,G) \from \LDPC[n,d,t,r]$ and $z \gets \F_2^n$. Output $(\sk = (H,z), \pk = (G,z))$.
        \item $\Encode(\secparam,(G,z))$: Sample $u \from \F_2^d$, $e \from \calS_{\eta n, n}$. Output $Gu \oplus z \oplus e$.
        \item $\Decode(\secparam,(H,z),x)$: If $\wt(H (x \oplus z)) < \left(\frac{1}{2} - \zeta\right) \cdot r$, output 1; otherwise output $\bot$.
    \end{itemize}
\end{construction}

The original construction of \cite{CG24} is pseudorandom under the subexponential LPN assumption.
Since ours is the same except that we used fixed-weight error, ours is pseudorandom under the polynomially related \emph{exact LPN} assumption:

\begin{assumption}[Exact LPN assumption \cite{JKPT12}] \label{assumption:xLPN}
    For $\eta \in (0, 1/2)$ and $g : \N \to \N$, the $\xLPN_{g,\eta}$ assumption states that for every $n \in \N$ and every polynomial-time adversary $\adv$,
    \[
        \abs{\Pr_{\substack{A \from \F_2^{n \times g(n)}\\ s \from \F_2^{g(n)}\\ e \from \calS_{\eta n, n}}}[\adv(A, As \oplus e) = 1] - \Pr_{\substack{A \from \F_2^{n \times g(n)}\\ u \from \F_2^n}}[\adv(A, u) = 1]} \leq \negl[n].
    \]
\end{assumption}

Similarly to standard LPN, the exact LPN assumption implies that any polynomial number of samples of the form $(A,As \oplus e)$ are indistinguishable from uniformly random samples.

Importantly, the exact LPN assumption is \emph{equivalent} to the standard LPN assumption.

\begin{fact}[Proposition 2.3 from \cite{JKPT12}]
    The hardness of $\xLPN_{g,\eta}$ is polynomially related to the hardness of $\LPN_{g,\eta}$ as defined in \cite{CG24}.
\end{fact}

\subsection{The toolkit} \label{subsection:toolkit}
The basic idea is to use the fact that there are actually $n^{\Omega(\log n)}$ parity checks that are satisfied by a random $\calC$, even though our detector uses just $r = O(n)$ of them.
Since the adversary only sees $\calC$, from their perspective the actual parity checks used by our detector are chosen uniformly at random from $P_{\calC,t}$.
This is formalized in \Cref{lemma:adv-knowledge}.

\begin{lemma} \label{lemma:adv-knowledge}
    Let $H \in \F_2^{r \times n}$ be a random matrix where each row is $t$-sparse. Let $\calC$ be a random $d$-dimensional subspace selected from $\ker H$. Then for all $d$-dimensional subspaces $\calC_* \subseteq \F_2^n$ and all $H_*, H_*' \in \F_2^{r \times n}$ such that $\calC_* \subseteq \ker H_* \cap \ker H_*'$,
    \[
        \Pr[H = H_* \mid \calC = \calC_*] = \Pr[H = H_*' \mid \calC = \calC_*].
    \]
\end{lemma}
\begin{proof}
    Since $\calC$ is selected uniformly at random from all $d$-dimensional subspaces of $\ker H$,
    \[
        \Pr[\calC = \calC_* \mid H = H_*] = \Pr[\calC = \calC_* \mid H = H_*'].
    \]
    Since $H$ is selected uniformly at random to begin with, we also have $\Pr[H=H_*] = \Pr[H=H_*']$. The result follows from Bayes' rule.
\end{proof}

\Cref{lemma:adv-knowledge} reduces the problem to showing that decoding with $P_{\calC,t}$ is adaptively robust --- even though $P_{\calC,t}$ contains far too many parity checks for the decoder to actually use.
The following lemma is the key ingredient to showing that decoding with $P_{\calC,t}$ is adaptively robust.
Recall from \Cref{section:prelims} that $\bias(z) = \frac{1}{m} \sum_{i=1}^m (-1)^{z_i}$.
Note that $\bias(z) = 1 - 2\wt(z)/n$ and $\wt(z) = (1/2 - \bias(z)/2) \cdot n$.
Recall that we use $W_t \in \F_2^{\binom{n}{t} \times n}$ to denote the matrix whose rows are all weight-$t$ parity checks.

\begin{lemma} \label{lemma:omar}
    Let $N_{\calC,t}$ be the number of $t$-sparse parities consistent with $\calC$. For any vector $x \in \F_2^n$ and code $\calC$ with $N_{\calC, t} > 0$, we have that
    \[
        \bias(P_{\calC, t} x) = \frac{\binom{n}{t} \cdot 2^{-\dim(\calC)}}{N_{\calC, t}} \sum_{c \in \calC} \bias(W_t (x \oplus c)).
    \]
\end{lemma}
\begin{proof}
    First, observe that for any vector $v \in \F_2^n$, we have that
    \begin{equation*}
        1\{v \in \calC^\perp\} = 2^{-\dim(\calC)}\sum_{c \in \calC}{(-1)^{v \cdot c}} \ .
    \end{equation*}
    Using this identity, we have the following series of equalities:
    \begin{align*}
        N_{\calC,t} \cdot \bias(P_{\calC,t}x) &= \sum_{v \in \text{rows}(P_{\calC, t})}{(-1)^{v \cdot x}} \\
        &= \sum_{v \in \text{rows}(W_t) \cap \calC^\perp}{(-1)^{v \cdot x}} \\
        &= \sum_{v \in \text{rows}(W_t)}{(-1)^{v \cdot x} 1\{v \in \calC^\perp\}} \\
        &= \sum_{v \in \text{rows}(W_t)}{(-1)^{v \cdot x} \left( 2^{-\dim(\calC)} \sum_{c \in \calC}{(-1)^{v \cdot c}} \right)} \\
        &= 2^{-\dim(\calC)} \sum_{c \in \calC}{\left(\sum_{v \in \text{rows}(W_t)}{(-1)^{v \cdot (x \oplus c)}}\right)} \\
        &= 2^{-\dim(\calC)} \sum_{c \in \calC}{\binom{n}{t} \bias(W_t(x \oplus c))} \ . \qedhere
    \end{align*}
\end{proof}

The quantity outside of the sum, $\binom{n}{t} \cdot 2^{-\dim(\calC)} / N_{\calC, t}$, can be assumed to be 1 by \Cref{fact:parity-concentration}.
Therefore \Cref{lemma:omar} allows us to reason about $\bias(P_{\calC, t})$, which appears to be a hopelessly unwieldy quantity, by reasoning instead about random $t$-sparse parities.

\subsection{An adaptively robust zero-bit PRC} \label{subsection:zero-bit}
We will show in \Cref{theorem:parity-bias-lower-bound} that, if $\calC$ is selected as the span of $d$ random vectors, then $\bias(P_{\calC,t} x) \ge n^{\Omega(1)-1/2}$ for \emph{every} $x \in \F_2^n$ with $\bias(x) = \Omega(1)$; by \cite[Lemma 9]{CG24}, this also holds when $\calC$ is selected by the key generation algorithm described above.
By \Cref{lemma:adv-knowledge}, it follows that any $x$ the adversary chooses based only on knowledge of $\calC$ will be detected with high probability over the choice of $r$ parities from $P_{\calC,t}$.

\begin{fact} \label{fact:parity-concentration}
    If $\omega(1) \le t \le n^{o(1)}$ and $d \le (1-\Omega(1)) \cdot t \log n$, then
    \[
        \Pr_{G \from \F_2^{n \times d}}\left[N_{\Im(G),t} = (1 \pm \negl[n]) \cdot \binom{n}{t} \cdot 2^{-d}\right] = 1 - \negl[n].
    \]
\end{fact}
\begin{proof}
    The proof is a simple application of Chebyshev's inequality, and is also shown in \cite[Lemma 11]{CG24}. We present it here for completeness.
    
    For uniformly random $G \from \F_2^{n \times d}$, let $X_w = 1\{w G = 0\}$ be random variables for each $w \in W_t$. Note that $N_{\Im(G), t} = \sum_{w \in W_t} X_w$, and for each $w \in W_t$, $\E[X_w] = 2^{-d}$ and $\Var[X_w] = 2^{-d} - 2^{-2d}$.

    Letting $\alpha = \binom{n}{t} \cdot 2^{-d}$,
    \[
        \E N_{\Im(G), t} = \sum_{w \in W_t} \E[X_w] = \alpha.
    \]
    Furthermore, $\{X_w\}_{w \in W_t}$ are pairwise independent. So by Chebyshev's inequality,\footnote{We use the following version of Chebyshev's inequality: for $N$ pairwise independent random variables $X_1, \dots, X_N$ where $\E X_1 = \cdots = \E X_N = \mu$ and $\Var X_1 = \cdots = \Var X_N = \sigma^2$, $\Pr[\abs{\sum_{i=1}^N X_i - N \mu} > \tau] \le N \sigma^2 / \tau^2$.}
    \begin{align*}
        \Pr[\abs{N_{\Im(G), t} - \alpha} > \tau] &\le \binom{n}{t} \cdot (2^{-d} - 2^{-2d}) / \tau^2 \\
        &\le \alpha / \tau^2.
    \end{align*}
    Let $\tau = \alpha^{2/3}$. We have
    \[
        \Pr[N_{\Im(G), t} = \left(1 \pm \alpha^{-1/3}\right) \alpha] \ge 1 - \alpha^{-1/3}.
    \]
    Invoking the assumptions that $t \le n^{o(1)}$ and $d \le (1-\Omega(1)) \cdot t \log n$,
    \begin{align*}
        \alpha &\ge \left(\frac{n}{t}\right)^t \cdot 2^{-d} \\
        &= 2^{t \log n - d - t \log t} \\
        &= 2^{t \log n - (1-\Omega(1)) \cdot t \log n - o(t \log n)} \\
        &= 2^{\Omega(t \log n)},
    \end{align*}
    which is super-polynomial by the assumption that $t = \omega(1)$. This completes the proof.
\end{proof}

\begin{fact} \label{fact:parity-bias}
    For any vector $x \in \F_2^n$ and any even $t \le \sqrt{n/2}$,
    \begin{equation} \label{eq:parity-bias}
        \bias(x)^t \cdot \left(1-\frac{2t^2/n}{\bias(x)^2}\right) \le \bias(W_t x) \le \bias(x)^t
    \end{equation}
\end{fact}
\begin{proof}
    Recall that $\bias(W_t x) = \E_{w \from W_t} (-1)^{v \cdot x}$.
    Let $W_t^*$ be the distribution over $\F_2^n$ defined by summing $t$ independent, random indicator vectors $e_i$.
    In other words, $W_t^*$ is the same distribution as $W_t$ except that we allow repeats. Let $Q_t = \Pr[\wt(w) = t \mid w \from W_t^*]$ be the probability that there are no repeated indices sampled under $W_t^*$. Note that $W_t^*$ conditioned on having no repeats is exactly $W_t$; and $W_t^*$ conditioned on having repeats is exactly $W_{t-2}^*$. We have
    \begin{align*}
        \bias(x)^t &= \E_{w \from W_t^*} (-1)^{w \cdot x} \\
        &= Q_t \E_{w \from W_t}[(-1)^{w \cdot x} \mid \wt(w)=t] + (1-Q_t) \E_{w \from W_t^*}[(-1)^{w \cdot x} \mid \wt(w) < t] \\
        &= Q_t \cdot \bias(W_t x) + (1-Q_t) \cdot \bias(x)^{t-2}.
    \end{align*}
    Rearranging and using the fact that $Q_t \ge 1-t^2/n$, we obtain the lower bound on $\bias(W_t x)$ as
    \begin{align*}
        \bias(W_t x) &= Q_t^{-1} \cdot \bias(x)^t + (1-Q_t^{-1}) \cdot \bias(x)^{t-2} \\
        &\ge \bias(x)^t + \left(1 - \frac{1}{1-t^2/n}\right) \cdot \bias(x)^{t-2} \\
        &\ge \bias(x)^t - \frac{2t^2}{n} \cdot \bias(x)^{t-2}.
        % &\ge \bias(x)^t - (1+o(1)) (t^2/n) \cdot \bias(x)^{t-2}.
    \end{align*}
    For the last inequality, we used that $t$ is even and $t \le \sqrt{n/2}$ and $1-1/(1-z) \ge -2z$ for $z \in [0,1/2]$. The upper bound on $\bias(W_t x)$ follows by maximizing $Q_t^{-1} \cdot \bias(x)^t + (1-Q_t^{-1}) \cdot \bias(x)^{t-2}$ over $Q_t \in [0,1]$, using the fact that $0 \le \bias(x)^t \le \bias(x)^{t-2}$.
\end{proof}

\begin{fact} \label{fact:parity-lower-bound}
    For any vector $x \in \F_2^n$ and any even $t = n^{o(1)}$,
    \[
        \bias(W_t x) \ge -n^{-(1/2-o(1)) \cdot t}.
    \]
\end{fact}
\begin{proof}
    If $\bias(x)^2 \ge 2t^2/n$, then \Cref{fact:parity-bias} implies that $\bias(W_t x) \ge 0$.

    If $\bias(x)^2 < 2t^2/n$, then by \Cref{fact:parity-bias} we still have that
    \begin{align*}
        \bias(W_t x) &\ge -\frac{2t^2}{n} \cdot \bias(x)^{t-2} \\
        &\ge -\frac{2t^2}{n} \cdot \left(\frac{2t^2}{n}\right)^{t/2-1} \\
        &= -\left(\frac{2t^2}{n}\right)^{t/2} \\
        &\ge -n^{-(1/2-o(1)) \cdot t}. \qedhere
    \end{align*}
\end{proof}

Note that \Cref{theorem:parity-bias-lower-bound} only applies for \emph{even} values of $t$.

\begin{theorem} \label{theorem:parity-bias-lower-bound}
    Let $t = n^{o(1)}$ be even and $d \le (1/2 - \Omega(1)) \cdot t \log n$. Let $G \from \F_2^{n \times d}$ be uniformly random and $\calC$ be the column span of $G$. Then with probability $1-\negl[n]$ over $G$,
    \[
        \bias(P_{\calC, t} e) \ge (1-o(1)) \cdot \bias(e)^t \text{ for all } e \in \F_2^n \text{ s.t. } \bias(e) = \Omega(1).
    \]
\end{theorem}
\begin{proof}
    By \Cref{lemma:omar} and \Cref{fact:parity-concentration}, it suffices to bound
    \[
        \sum_{c \in \calC} \bias(W_t (x \oplus c)) = \bias(W_t x) + \sum_{c \in \calC \setminus \{0\}} \bias(W_t(x \oplus c)).
    \]
    We bound the first and second terms separately. For the first term, \Cref{fact:parity-bias} implies that that for any $x$ such that $\bias(x) = \Omega(1)$,
    \begin{align*}
        \bias(W_t x) &\ge \bias(x)^t \cdot \left(1 - \frac{2t^2/n}{\bias(x)^2}\right) \\
        &= (1-O(t^2/n)) \cdot \bias(x)^t.
    \end{align*}
    For the second term, \Cref{fact:parity-lower-bound} and the assumption that $d \le (1/2 - \Omega(1)) \cdot t \log n$ imply that
    \begin{align*}
        \sum_{c \in \calC \setminus \{0\}} \bias(W_t(x \oplus c)) &\ge -2^d \cdot n^{-(1/2-o(1)) \cdot t} \\
        &\ge -n^{(1/2 - \Omega(1)) \cdot t} \cdot n^{-(1/2-o(1)) \cdot t} \\
        &= -n^{-\Omega(t)}. \qedhere
    \end{align*}
\end{proof}

\begin{corollary}
    For any $\delta \in (0,1/2)$ and $r=n^{\Omega(1)}$, there exist $\eta = \Omega(1)$, $t = \Theta(\log n)$, and $d = \Omega(\log^2 n)$ such that the zero-bit public-key pseudorandom code $\LDPCPRC_0[n, d, t, r, \eta, r^{-1/4}]$ is adaptively $\delta$-robust (\Cref{definition:adaptive-robustness-pk}).
\end{corollary}
\begin{proof}
    Let
    \begin{itemize}
        \item $\eta = 1/4-\delta/2$,
        \item $t = \log(4r^{-1/4}) / \log(1/2-\delta)$, and
        \item $d = (1/3) \cdot t \log n$.
    \end{itemize}
    In order to invoke \Cref{fact:parity-bias}, $t$ must be even, so we assume for simplicity that $4\log(4r^{-1/4}) / \log(1/2-\delta)$ is an even integer.
    
    Let $m', c, x, r$ be defined as in $\Grobustpk$ and suppose that $\wt(x \oplus c) \le \delta n$.
    Since $\LDPCPRC_0$ is a zero-bit PRC, we omit $m$ from $\Grobustpk$.
    We do not omit $m'$, however, because it could still potentially be $\bot$.
    
    Let $c^* = c \oplus e^*$ where $e^*$ is the $\eta n$-sparse noise added by $\PRC.\Encode(\pk; r)$.
    Let $e = x \oplus c^*$ be the combined error $c \oplus c^*$ and $x \oplus c$.
    By our choice of $\eta$ and the assumption that $\wt(x \oplus c) \le \delta n$,
    \begin{align*}
        \wt(e) &\le \wt(x \oplus c) + \wt(c \oplus c^*) \\
        &\le \delta n + \eta n \\
        &= \left(1/4 + \delta/2\right) \cdot n
    \end{align*}
    Equivalently, this means that $\bias(e) \ge 1/2 - \delta = \Omega(1)$.
    By \Cref{theorem:parity-bias-lower-bound}, we have that
    \begin{align*}
        \bias(P_{\calC,t} e) &\ge (1-o(1)) \cdot \bias(e)^t \\
        &\ge (1-o(1)) \cdot (1/2-\delta)^t \\
        &= (1-o(1)) \cdot 4 r^{-1/4},
    \end{align*}
    so $\Pr_{w \from P_{\calC,t}}[w \cdot e = 1] \le 1/2 - 2r^{-1/4}$.

    Now since $e$ is computed as a function of the public key alone, \Cref{lemma:adv-knowledge} implies that we can view $H$ as being sampled after $e$.
    Therefore each parity check $w$ in $H$ has $\Pr_{w \from P_{\calC,t}}[w \cdot e = 1] \le 1/2 - 2r^{-1/4}$, and by a Chernoff bound over the parity checks in $H$,
    \[
        \Pr[\wt(H e) > (1/2 - r^{-1/4}) \cdot r] \le \negl[n].
    \]
    This means that
    \[
        \Pr[\PRC.\Decode(\sk; x) = \bot] \le \negl[n],
    \]
    completing the proof.
\end{proof}

\subsection{An adaptively robust single-bit PRC} \label{subsection:single-bit}

Let us now state our single-bit public-key PRC construction.
It is essentially a combination of two zero-bit public-key PRCs, used to encode 0 and 1 separately.
The decoder checks whether \emph{exactly one} of the zero-bit detectors accepts, in which case our decoder outputs the corresponding bit.
While this construction is generic in the underlying zero-bit PRC, our proof of adaptive robustness relies on the particular structure of the LDPC-based PRCs.

\begin{construction}[Single-bit public-key pseudorandom LDPC code, {$\LDPCPRC_1[n,d,t,r,\eta,\zeta]$}] \label{const:single-bit-ldpc-prc}
    Let $n, d, t, r : \mathbb{N} \to \mathbb{N}$ and $\eta, \zeta : \N \to [0,1/2)$ be efficiently-computable functions of the security parameter. We define $\LDPCPRC_0[n,d,t,r,\eta,\zeta]$ by the following algorithms, where we leave the dependence of $n,d,t,r,\eta,\zeta$ on $\secpar$ implicit:
    \begin{itemize}
        \item $\KeyGen(\secparam)$: Sample $(H_0,G_0), (H_1, G_1) \from \LDPC[n,d,t,r]$ and $z \gets \F_2^n$. Output $(\sk = (H_0,H_1,z), \pk = (G_0,G_1,z))$.
        \item $\Encode(\secparam,(G_0,G_1,z),m)$: Sample $u \from \F_2^d \setminus \{0\}$, $e \from \calS_{\eta n, n}$. Output $G_m u \oplus z \oplus e$.
        \item $\Decode(\secparam,(H_0,H_1,z),x)$: For $m \in \{0,1\}$, if $\wt(H_m (x \oplus z)) < \left(\frac{1}{2} - \zeta\right) \cdot r$ and $\wt(H_{1-m} (x \oplus z)) > \left(\frac{1}{2} - \zeta\right) \cdot r$, output $m$; otherwise output $\bot$.
    \end{itemize}
\end{construction}

In order to prove adaptive robustness of \Cref{const:single-bit-ldpc-prc}, we need to show two things.
First, that the zero-bit decoders are robust up to errors of weight $\delta n < n/4$.
This follows from the same argument as in \Cref{subsection:zero-bit} (although with different parameters).
Second, we need to show that the zero-bit decoders will \emph{not} detect any codewords that are within $\delta n$ Hamming distance of a codeword from the other zero-bit code.
This is crucial for our decoder to know which of the zero-bit PRCs was used.
We prove this second claim in \Cref{theorem:parity-bias-upper-bound}, using our special sauce \Cref{lemma:omar}, the Johnson bound, and the fact that random linear codes are approximately unbiased.

\begin{lemma}[Random linear codes are $\varepsilon$-balanced] \label{lemma:rlcs-balanced}
    Let $\calC \subseteq \F_2^n$ be a random linear code of dimension $d$. Then with probability $1-\negl[d]$, $\calC$ is $\varepsilon$-biased for $\varepsilon = 2\sqrt{d/n}$, i.e.,
    \[
        \abs{\bias(c)} \le 2\sqrt{\frac{d}{n}} \text{ for all } c \in \calC \setminus \{0\}.
    \]
\end{lemma}
\begin{proof}
    Let $G$ be the generator matrix for $\calC$.
    For any $z \in \F_2^d \setminus \{0\}$,
    \[
        \Pr_{G}[\abs{\bias(Gz)} > \varepsilon] \le 2 e^{-\varepsilon^2 n / 2}
    \]
    by a Chernoff bound. By a union bound over $z$,
    \[
        \Pr_{G}[\exists z \in \F_2^d \setminus \{0\} \text{ s.t. } \abs{\bias(Gz)} > \varepsilon] \le 2^{d - \varepsilon^2 n / 2 + 1} = \negl[d]
    \]
    if $\varepsilon = 2\sqrt{d/n}$.
\end{proof}

\begin{fact} \label{fact:rlc-closest}
    With probability $1-\negl[n]$ over random linear codes $\calC_0, \calC_1$ of dimensions $d_0, d_1$ where $\omega(\log n) \le d_0, d_1 \le o(n)$, for every $e$ we have
    \[
        \max_{\substack{c_0 \in \calC_0 \\ c_1 \in \calC_1 \setminus \{0\}}} \abs{\bias(c_0 \oplus c_1 \oplus e)} \le 1 - \bias(e) + 2 \sqrt{(d_0+d_1)/n}.
    \]
\end{fact}
\begin{proof}
    First we use the fact that $\abs{\wt(c_0 \oplus c_1 \oplus e) - \wt(c_0 \oplus c_1)} \le \wt(e)$ to see that
    \begin{align*}
        \max_{\substack{c_0 \in \calC_0 \\ c_1 \in \calC_1 \setminus \{0\}}} \abs{\bias(c_0 \oplus c_1 \oplus e)} &= \max_{c \in \calC} \abs{1-2\wt(c_0 \oplus c_1 \oplus e)/n} \\
        &\le 2\wt(e)/n + \max_{\substack{c_0 \in \calC_0 \\ c_1 \in \calC_1 \setminus \{0\}}} \abs{1-2\wt(c_0 \oplus c_1)/n} \\
        &= 1-\bias(e) + \max_{\substack{c_0 \in \calC_0 \\ c_1 \in \calC_1 \setminus \{0\}}} \abs{\bias(c_0 \oplus c_1)}.
    \end{align*}
    Let $\calC = \{c_0 \oplus c_1 \mid c_0 \in \calC_0, c_1 \in \calC_1\}$.
    With probability $1-\negl[n]$, $\calC_0$ and $\calC_1$ are linearly independent, in which case $\calC$ is a random linear code of dimension $d_0+d_1$.
    So by \Cref{lemma:rlcs-balanced},
    \begin{align*}
        \max_{\substack{c_0 \in \calC_0 \\ c_1 \in \calC_1 \setminus \{0\}}} \abs{\bias(c_0 \oplus c_1)} &\le \max_{c \in \calC \setminus \{0\}} \abs{\bias(c)} \\
        &\le 2 \sqrt{(d_0+d_1)/n}
    \end{align*}
    with probability $1-\negl[n]$ over $\calC_0, \calC_1$, completing the proof.
\end{proof}

\begin{lemma}[Johnson bound adapted from {\cite[Equation 7.6]{GRS23}}] \label{lemma:johnson}
    Let $\calC \subseteq \F_2^n$ be a code of distance at least $(1-\delta) n/2$. Then for any $x \in \F_2^n$ and any $\tau > \sqrt{\delta}$,
    \[
        \abs{\{c \in \calC : \abs{\bias(x \oplus c)} \ge \tau\}} \le \frac{1-\delta}{\tau^2 - \delta}.
    \]
\end{lemma}

\begin{theorem} \label{theorem:parity-bias-upper-bound}
    Suppose that $\omega(1) \le t \le n^{o(1)}$ and $d = \varepsilon t \log n$ for some constant $\varepsilon \in (0,1/8)$. Let $G, G' \from \F_2^{n \times d}$ be uniformly random and $\calC, \calC'$ be the column spans of $G, G'$.
    Then with probability $1-\negl[n]$ over $G, G'$,
    \[
        \bias(P_{\calC, t} (c' \oplus e)) \le (1+o(1)) \cdot n^{4\varepsilon} \cdot \left(1 - \bias(e) + \sqrt{8d/n}\right)^t \text{ for all } c' \in \calC' \setminus \{0\}, e \in \F_2^n.
    \]
\end{theorem}
\begin{proof}
    By \Cref{lemma:omar} (with \Cref{fact:parity-concentration,fact:parity-bias} applied), it suffices to bound
    \[
        \sum_{c \in \calC} \bias(c' \oplus e \oplus c)^t = \sum_{\substack{c \in \calC \\ \abs{\bias(c' \oplus e \oplus c)} \ge n^{-2 \varepsilon}}} \bias(c' \oplus e \oplus c)^t + \sum_{\substack{c \in \calC \\ \abs{\bias(c' \oplus e \oplus c)} < n^{-2 \varepsilon}}} \bias(c' \oplus e \oplus c)^t.
    \]
    Since $\dim \calC \le d$, the term on the right is at most $2^d \cdot n^{-2 \varepsilon t}$.
    
    By \Cref{lemma:johnson,lemma:rlcs-balanced} and \Cref{fact:rlc-closest}, with probability $1 - \negl[n]$ the term on the left is at most
    \begin{align*}
        \left(\frac{1-2\sqrt{d/n}}{n^{-4 \varepsilon}-2\sqrt{d/n}}\right) \cdot \max_{c \in \calC} \bias(c' \oplus e \oplus c)^t &\le \left(\frac{1}{n^{-4 \varepsilon}-2\sqrt{d/n}}\right) \cdot \left(1 - \bias(e) + \sqrt{8d/n}\right)^t \\
        &\le (1+o(1)) \cdot n^{4\varepsilon} \cdot \left(1 - \bias(e) + \sqrt{8d/n}\right)^t
    \end{align*}
    for all $c' \in \calC' \setminus \{0\}$ and $e \in \F_2^n$.
\end{proof}

\begin{corollary}
    For any $\delta \in (0,1/4)$ and $r = n^{\Omega(1)}$, there exist $\eta = \Omega(1)$, $t = \Theta(\log n)$, and $d = \Omega(\log^2 n)$ such that the single-bit public-key pseudorandom code $\LDPCPRC_1[n, d, t, r, \eta, 3r^{-1/5}/2]$ is adaptively $\delta$-robust (\Cref{definition:adaptive-robustness-pk}).
\end{corollary}
\begin{proof}
    Let
    \begin{itemize}
        \item $\eta = 1/8-\delta/2$,
        \item $t = (1/5) \cdot \log r$,
        \item $\delta' = \log(1/(1/4+\delta))-1$,
        \item $\varepsilon = \min\{t \delta' / 4 \log n, 1/17\}$, and
        \item $d = \varepsilon t \log n$.
    \end{itemize}
    
    Let $\calC_0, \calC_1$ be the column spans of $G_0, G_1$, and $z$ be the one-time pad, sampled from $\LDPCPRC_1$.
    Let $m, r, c, x, m'$ be defined as in $\Grobustpk$ and suppose that $\wt(x \oplus c) \le \delta n$.

    We need to show that $x$ is detected under $H_m$ and not detected under $H_{1-m}$.
    Recall from \Cref{const:single-bit-ldpc-prc} that this means that $\bias(H_m (x \oplus z)) > 3r^{-1/5}$ and $\bias(H_{1-m} (x \oplus z)) < 3r^{-1/5}$.
    
    We will prove these two inequalities separately, but both cases will center on the total error $e$ induced by the adversary (including both the $\eta n$-sparse encoding noise, selected via $r$, and the $\delta n$-sparse attack).
    Let $c^* \in \calC_m \setminus \{0\}$ and $e$ be such that
    \[
        x = c^* \oplus e \oplus z
    \]
    and $\wt(e) \le (\delta + \eta) \cdot n = \left(1/8 + \delta/2\right) \cdot n$ (or equivalently, $\bias(e) \ge 3/4 - \delta$).

    \paragraph{Proof that $\bias(H_m (x \oplus z)) > 3r^{-1/5}$.}
    This part is essentially the same as \Cref{subsection:zero-bit}, except that we have set $t$ to be larger because we do not want to detect beyond $1/4$ errors anymore.
    
    By \Cref{theorem:parity-bias-lower-bound}, we have that with probability $1-\negl[n]$,
    \begin{align*}
        \bias(P_{\calC_m,t} (x \oplus z)) &= \bias(P_{\calC_m,t} e) \\
        &\ge (1-o(1)) \cdot \bias(e)^t \\
        &\ge (1-o(1)) \cdot (3/4-\delta)^t \\
        &= (1-o(1)) \cdot 2^{t \log(3/4-\delta)} \\
        &= r^{\log(3/4-\delta)/5} \\
        &\ge 4 r^{-1/5}
    \end{align*}
    because $\log(3/4-\delta) > \log(1/2) = -1$. By \Cref{lemma:adv-knowledge} and a Chernoff bound, it follows that
    \[
        \Pr[\bias(H_m (x \oplus z)) \le 3r^{-1/5}] \le \negl[n].
    \]

    \paragraph{Proof that $\bias(H_{1-m} (x \oplus z)) < 3r^{-1/5}$.}
    By \Cref{theorem:parity-bias-upper-bound}, with probability $1-\negl[n]$,
    \begin{align*}
        \bias(P_{\calC_{1-m},t} (x \oplus z)) &= \bias(P_{\calC_{1-m},t} (c^* \oplus e)) \\
        &\le (1+o(1)) \cdot n^{4\varepsilon} \cdot \left(1/4 + \delta + \sqrt{8d/n}\right)^t \\
        &\le (1+o(1)) \cdot n^{4\varepsilon} \cdot \left[(1/4 + \delta)^t + t \sqrt{8d/n}\right] \\
        &= (1+o(1)) \cdot \left[2^{4\varepsilon \log n - t \log(1/(1/4+\delta))} + n^{4\varepsilon} \cdot t \sqrt{8d/n}\right] \\
        &= (1+o(1)) \cdot \left[2^{4\varepsilon \log n - t \delta' - t} + n^{4\varepsilon} \cdot t \sqrt{8d/n}\right].
    \end{align*}
    Recall that $\delta' = \log(1/(1/4+\delta))-1$, which is positive since $\delta < 1/4$, and $\varepsilon = \min\{t \delta' / 4 \log n, 1/17\}$.
    Using $\varepsilon \le t \delta' / 4 \log n$ for the left term and $\varepsilon \le 1/17$ for the right, we have
    \begin{align*}
        \bias(P_{\calC_{1-m},t} (x \oplus z)) &\le (1+o(1)) \cdot \left[2^{-t} + r^{-1/4} \cdot t \sqrt{8d}\right] \\
        &\le 2 r^{-1/5}.
    \end{align*}
    By \Cref{lemma:adv-knowledge} and a Chernoff bound, it follows that
    \[
        \Pr[\bias(H_{1-m} (x \oplus z)) \ge 3r^{-1/5}] \le \negl[n]. \qedhere
    \]
\end{proof}

\section{Ideal PRCs: the secret-key setting} \label{section:ideal-security}

This section is dedicated to extending the notion of {\em secret-key} adaptively secure PRC to what we call {\em Ideal PRC}. The corresponding definition is given in Section~\ref{sec:def-ideal}. We then show how to {\em generically} 
build such a PRC from {\em any} (secret-key) single-bit adaptively secure PRC --- such as the one from 
Section~\ref{subsection:single-bit} --- in two-steps. First, in Section~\ref{subsection:rate-boost-sk} we 
generically show how to boost the information rate of adaptively-secure PRC to become essentially optimal, almost without sacrificing the robustness. Then, in Section~\ref{subsection:ideal-security-proof}, we (also generically) show how to upgrade the secret-key adaptive security from Section~\ref{sec:def-adaptive} to our notion of ideal security in Section~\ref{sec:def-ideal}, using any PRF (which is implied by the existence of PRC, and thus does not require a new assumption). Moreover, this transformation almost preserves the information rate and the robustness, resulting in nearly optimal-rate ideal PRCs. 

\subsection{Definition}\label{sec:def-ideal}
Consider the following security games. The goal of the adversary is to determine whether it is given oracle access to the actual detector of the PRC, or to an ideal decoding oracle.

\noindent \underline{$\Greal_{\adv, \delta, \PRC}(1^\secpar)$:}
\begin{enumerate}
    \item The challenger samples $\sk \gets \PRC.\KeyGen(1^\secpar)$.
    \item The adversary is allowed to make encoding and decoding queries:
    \begin{itemize}
        \item For each encoding query $m$, the challenger responds with $\PRC.\Encode(\sk, m)$.
        \item For each decoding query $x$, the challenger responds with $\PRC.\Decode(\sk, x)$.
    \end{itemize}
    \item The adversary returns a bit $\hat{b}$.
\end{enumerate}

\noindent \underline{$\Gideal_{\adv, \delta}(1^\secpar)$:}
\begin{enumerate}
    \item The challenger sets $\transcript = \emptyset$.
    \item The adversary is allowed to make encoding and decoding queries:
    \begin{itemize}
        \item For each encoding query $m$, the challenger responds with a fresh random string $c$ and sets $\transcript = \transcript \cup \{(m, c)\}$.
        \item For each decoding query $x$, the challenger responds with a random $m$ such that $(m, c) \in \transcript$ and $\wt(c \oplus x) \leq \delta n$; if no such $m$ exists, the challenger responds with $\bot$.
    \end{itemize}
    \item The adversary returns a bit $\hat{b}$.
\end{enumerate}

\begin{definition}[Ideal secret-key PRC security] \label{definition:ideal-security-sk}
    We say that a secret-key pseudorandom code $\PRC$ satisfies \emph{ideal $\delta$-robust security} if, for any efficient adversary $\adv$,
    $$\abs{\Pr[\Greal_{\adv,\delta,\PRC}(1^\secpar) = 1] - \Pr[\Gideal_{\adv, \delta}(1^\secpar) = 1]} \leq \negl.$$
\end{definition}

\subsection{Boosting the information rate}
\label{subsection:rate-boost-sk}

We recall the linear-rate PRC from \cite{CG24}:

\begin{construction}[Linear-rate public-key PRC] \label{const:linear-rate-prcs}
    Let $\PRC_1$ be a single-bit public-key PRC with block length $\secpar$. Let $(\Enc, \Dec)$ be any error-correcting code with block length $n \ge \secpar$ and messages of length $k$. Let $\PRG : \{0,1\}^\secpar \to \{0,1\}^n$ be any pseudorandom generator. We define $\PRC_k[\PRC_1, (\Enc, \Dec), \PRG]$ which is a $k$-bit public-key PRC as follows:
    \begin{itemize}
        \item $\KeyGen_k(\secparam)$: Sample $\sk' \gets \PRC_1.\KeyGen(1^\secpar)$ and a random permutation $\pi : [\secpar^2 + n] \to [\secpar^2 + n]$. Output $\sk = (\sk', \pi)$.
        \item $\Encode_k(\sk,\m)$: Given as input a message $\m \in \{0,1\}^k$, let $s \gets \{0,1\}^\secpar$, $r \gets \PRG(s)$, and
        \[
            x \gets \PRC_1.\Encode(\pk, r_1) || \ldots || \PRC_1.\Encode(\pk, r_\secpar) || \PRG(r) \oplus \Enc(\m).
        \]
        Output $\PermBits(x, \pi)$.
        \item $\Decode_k(\sk, c')$: Let $c = \PermBits(c', \pi^{-1})$.
        Parse $c = c_1 || \ldots || c_{\secpar} || c_{\secpar+1}$ as $\secpar$ length-$\secpar$ blocks followed by a length-$n$ block.
        The decoder then computes $y_i = \PRC_1.\Decode(\sk, y_i)$ for all $i \in [\secpar]$ and lets $y = y_1 || \ldots || y_{\secpar}$.
        It then outputs $\m \gets \Dec(\PRG(y) \oplus c_{\secpar+1})$.
    \end{itemize}
\end{construction}

Observe that asymptotically as the message length increases, the information rate of $\PRC_k$ approaches that of the underlying error-correcting code $(\Enc, \Dec)$.
There exist binary, linear-rate error-correcting codes that are robust to any rate of worst-case errors $\alpha \in (0, 1/4)$ --- for instance, see \cite{JST21}.

\begin{remark*}
    This PRC is public-key in the sense that codewords are pseudorandom even to an adversary that knows the public key. 
    However, it is only \emph{secret-key} adaptively robust, to an adversary that does \emph{not} know the public key.
    Later, in \Cref{subsection:rate-boost-pk}, we construct a linear-rate PRC that is public-key adaptively robust.
\end{remark*}

\begin{theorem} \label{theorem:multi-bit}
    If $\PRC_1$ and $(\Enc, \Dec)$ are both $\alpha$ secret-key adaptively robust, then $\PRC_k[\PRC_1, (\Enc, \Dec), \PRG]$ is $\alpha - o(1)$ secret-key adaptively robust.
\end{theorem}
\begin{proof}
    First, observe that by pseudorandomness of $\PRC_1$ we can replace the bits $r_1, \ldots, r_\secpar$ encoded under $\PRC_1$ with $0^\secpar$.
    Then, by pseudorandomness of $\PRG$, we can replace $\PRG(r)$ with a uniformly random string.

    Let $\tilde{c}$ denote the output of $\PRC_k.\Encode(\sk,\m)$ with up to an $\alpha - \epsilon$ fraction of adversarial substitutions applied, for any constant $\epsilon > 0$.
    Recall that $\tilde{c}$ consists of a permutation applied to $\secpar+1$ blocks each of length at least $\secpar$.
    We will show that the adversary cannot find errors that concentrate on any block; we then complete the proof by invoking the adaptive robustness of $\PRC_1$ and $(\Enc, \Dec)$.

    % We first argue that if $\Encode_1$ is replaced with an oracle that responds to each query with a fresh uniform string, the adversary's errors will be roughly evenly distributed across the blocks.
    % Observe that if we replace $\Encode_1$ with this uniform oracle, the outputs of $\PRC_k.\Encode$ reveal no information about $\pi$; therefore, $\pi$ can equivalently be drawn after the adversary has submitted $x$.

    For any $c$ such that $(m,c) \in \transcript$, let $e = x \oplus c$.
    $\PermBits(e, \pi^{-1})$ consists of $\ell+1$ blocks of length at least $\secpar$.
    If the permutation were to be chosen independently of $e$, the number of errors in block $i \in [\ell]$ is a random variable $X_i \sim \Hyp(\ell n' + n, (\alpha - \epsilon)(\ell n' + n), n)$, following the hypergeometric distribution. 
    The number of errors in the last block would be distributed as $X_{\ell+1} \sim \Hyp(\ell n' + n, (\alpha - \epsilon)(\ell n' + n), n')$.
    By \Cref{theorem:hyper}, for all $i \in [\ell+1]$,
    $$\Pr\left[X_i \geq \alpha n\right] \leq e^{-2\epsilon^2 n} = \negl.$$
    By a union bound, every block would have less than an $\alpha$ fraction of errors with overwhelming probability.

    We'll next show that if an adversary can find an error that concentrates on some block more than the above bounds, one can break pseudorandomness.
    
    Consider an adversary $\adv$ that is given oracle access to $\Encode_k(\sk, \cdot)$, and returns an error vector $e$ that has at least an $\alpha$ fraction of errors on some block with non-negligible probability.
    $\adv$ can be used to construct a distinguisher $\bdv$ breaking pseudorandomness of $\PRC_k$ as follows.
    $\bdv$ chooses a random permutation $\pi$. When $\adv$ queries $\m$ to $\Encode_k(\sk, \cdot)$, $\bdv$ forwards the corresponding queries $\m_i$ to its oracle (which is either $\Encode_1$ or the uniform distribution) to obtain responses $c_i$.
    It then lets $c = c_1|| \ldots || c_{\ell+1}$ and returns $c' = \PermBits(c, \pi)$.
    When $\bdv$ obtains an error vector $e$ from $\adv$, it inverts the permutation and checks if $e$ has at least an $\alpha$ fraction of errors on any block; if so, it knows it received PRC codewords rather than true randomness.

    We've thus shown that there are less than an $\alpha$ fraction of errors on every block.
    Furthermore, these errors are chosen by an efficient adversary, since $\adv$ and $\bdv$ together run in polynomial time.
    Therefore, adaptive robustness of $\PRC_1$ implies that $\Decode_1$ recovers every bit of $r$ correctly.
    Adaptive robustness of the error-correcting code implies that $\m$ is recovered by $\Dec$.
\end{proof}

\subsection{The non-malleable transformation and ideal security}
\label{subsection:ideal-security-proof}
Let $\PRC$ be a $k$-bit pseudorandom code with adaptive robustness up to $\delta$ errors. 
Let $F: \{0,1\}^* \to \{0,1\}^\secpar \times \{0,1\}^\secpar$ be a PRF.
Then there exists a secret-key PRC, say $\PRCsharp$, satisfying $\delta$-ideal security.

\noindent $\PRCsharp[\PRF, \PRC, \delta] = (\KeyGen^{sharp}, \Encode^{sharp}, \Decode^{sharp})$:
\begin{itemize}
    \item $\KeyGen^{sharp}(1^\secpar)$: Sample $\PRF.\sk \gets \PRF.\KeyGen(\secparam)$ and $\PRC.\sk \gets \PRC.\KeyGen(1^\secpar)$. Output $\sk = (\PRF.\sk, \PRC.\sk)$.
    
    \item $\Encode^{sharp}(\sk, \m)$: Sample $r \from \{0,1\}^\secpar$ and let $(R_1, R_2) = F_{\PRF.\sk}(r||\m)$.
     
    Output $\PRC.\Encode(\PRC.\sk, r||\m||R_2; R_1)$.
    
    \item $\Decode^{sharp}(\sk, c)$: Compute $r||\m||R_2 = \PRC.\Decode(\PRC.\sk, c)$ and let $(R_1, \tilde{R}_2) = F_{\PRF.\sk}(r||m)$. If
    \[
        \tilde{R}_2 = R_2 \text{ and }\wt(\PRC.\Encode(\PRC.\sk, r||\m||R_2; R_1) \oplus c) \le \delta n,
    \]
    output $\m$; otherwise, output $\bot$.
\end{itemize}

Now $\PRCsharp$ behaves almost identically to $\PRC$, except that it has a sharp decoding threshold. That is, any string that is more than $\delta$-far from a codeword will always be rejected. This is important because it allows our simulator to simply reject all inputs that are more than $\delta$ far from any of the codewords output so far.

\begin{theorem} \label{theorem:ideal-security-sk}
    Let $\PRC$ be any $k$-bit pseudorandom code that is $\delta$ adaptively robust, and let $\PRF$ be any pseudorandom function. Then $\PRCsharp[\PRF, \PRC, \delta]$ satisfies $\delta$ ideal security.
\end{theorem}
\begin{proof}[Proof Sketch]
    We use a hybrid argument. Suppose that the adversary makes at most $T$ queries to the decoding oracle. For $i \in 0, \dots, T$, define the following hybrid game.
    Let $\Encode^f$ and $\Decode^f$ denote the encoding and decoding algorithms of $\PRCsharp$, where its PRF is replaced with a random function $f$.
    
    \noindent \underline{$\calH^i_{\adv, \PRCsharp, \delta}(1^\secpar)$:}
    \begin{enumerate}
        \item The challenger samples $\sk \gets \PRCsharp.\KeyGen(1^\secpar)$
        \item The challenger sets $q = 1$ and $\transcript = \emptyset$.
        \item The adversary is allowed to make encoding and decoding queries.
        \begin{itemize}
            \item For each encoding query $m$:
            \begin{enumerate}
                \item The challenger computes $c \gets \PRCsharp.\Encode^f(\sk, m)$.
                \item The challenger lets $\transcript = \transcript \cup \{(m, c)\}$.
                \item The challenger responds with $c$.
            \end{enumerate}
            \item For each decoding query $x$:
            \begin{enumerate}
                \item If $q \le i$, the challenger responds with a random $m$ such that $(m, c) \in \transcript$ and $\wt(c \oplus x) \leq \delta n$; if no such $m$ exists, the challenger responds with $\bot$.
                \item If $q > i$, the challenger responds with $\PRCsharp.\Decode^f(\sk, x)$.
                \item The challenger increments $q$, setting $q = q + 1$.
            \end{enumerate}
        \end{itemize}
        \item The adversary returns a bit $\hat{b}$.
    \end{enumerate}

\paragraph{$\Greal to \calH^0$.} Observe that $\calH^0$ is identical to $\Greal$, but with the PRF replaced with a random function.
This is indistinguishable from $\Greal$ by security of the PRF.

\paragraph{$\calH^{i-1}$ to $\calH^i$.} For each $i \in [T]$, we show that the adversary's views in hybrids $\calH^{i-1}$ and $\calH^i$ are statistically close by showing that the response to the adversary's $i$th decoding query is the same in both cases.

Let $x$ be the adversary's $i$th decoding query, and let $\transcript$ be as defined in $\calH^i$ at the time that the adversary makes the $i$th decoding query.
We show consider two cases, showing that in both cases the response in $\calH^{i}$ is the same as in $\calH^{i-1}$ with high probability:
\begin{enumerate}
    \item There exists $(m,c) \in \transcript$ such that $\wt(c \oplus x) \leq \delta n$. In this case, with overwhelming probability, there is a unique such $m$ and $\PRCsharp.\Decode^f(\sk, x) = m$. This follows from adaptive robustness of $\PRCsharp$.
    \item There is no $(m,c) \in \transcript$ such that $\wt(c \oplus x) \leq \delta n$. In this case, because of the sharp decoding threshold of $\PRCsharp$, we have $\PRCsharp.\Decode^f(\sk,x) = \bot$ with overwhelming probability.
\end{enumerate}

\paragraph{Proof of (1).} 
Let $x_i$ be the adversary's $i^{\text{th}}$ decoding query, and consider the challenger's behavior in $\calH^i$ in the case that $\wt(c \oplus x_i) \leq \delta n$ for some $(m,c) \in \transcript$.
Let $m' = \PRCsharp.\Decode^f(\sk, x_i)$ be the challenger's response, and observe that by adaptive robustness, $m' = m$.
Note that we can invoke adaptive robustness (where the adversary is only given a single decoding query) because the responses to the decoding queries before $i$ can be computed from the transcript alone.

In $\calH^{i-1}$, the challenger responds with a random $m_j$ of the $(m_j, c_j)$ in the transcript such that $\wt(c_j \oplus x_i) \leq \delta n$.
By the above argument, $m_j = m'$ for all such $j$, and therefore the challenger responds with $m'$ in both $\calH^i$ and $\calH^{i-1}$.

\paragraph{Proof of (2).}
Let $x_i$ be the adversary's $i^{\text{th}}$ decoding query, and consider the challenger's behavior in $\calH^i$ in the case that $\wt(c \oplus x_i) > \delta n$ for all $(m,c) \in \transcript$.
The challenger first computes $r_i||m_i||R_2 = \PRCsharp.\Decode^f(\PRC.\sk, x_i)$.
It then computes $(R_1, \tilde{R_2}) \gets f(r_i||m_i)$.
Suppose for the sake of contradiction the challenger does not output $\bot$; this implies that $\tilde{R_2} = R_2$.
Furthermore, letting $\tilde{x_i} = \PRC.\Encode(\PRC.\sk, r_i||m_i||R_2; R_1)$, we have that $\wt(\tilde{x_i} \oplus c) \leq \delta n$.
If $r_i||m_i$ has not been queried to $f$ yet, $\tilde{R_2} = R_2$ holds with only negligible probability.
Therefore, $r_i||m_i$ must have been queried to $f$, and the challenger must have made this query in responding to one of the adversary's encoding queries.
That is, the challenger must have added to the transcript $c = \PRC.\Encode(\PRC.\sk, r_i||m_i||R_2; R_1)$ where $(R_1, R_2) = f(r_i||m_i)$.
Observe that $c = \tilde{x_i}$, which is within distance $\delta n$ of $x_i$; we have thus arrived at a contradiction.

Therefore, in $\calH^i$ the challenger responds with $\bot$ if $x_i$ is far from all codewords in the transcript, which is identical to the challenger's behavior in $\calH^{i-1}$.

\paragraph{$\calH^T$ to $\Gideal$.} 
Observe that $\calH^T$ is identical to $\Gideal$, except that in $\calH^T$ the challenger responds to encoding queries according to $\PRCsharp.\Encode^f$.
Observe that $\PRCsharp.\Encode^f$ is identical to $\PRC.\Encode$, provided that $\PRCsharp.\Encode^f$ never samples the same $r$ more than once.
As some $r$ is repeated with only negligible probability, $\calH^T$ is indistinguishable from a hybrid ${\calH^T}'$ where $\PRCsharp.\Encode^f$ is replaced by $\PRC.\Encode$.
Now, observe that the challenger in ${\calH^T}'$ makes only encoding (and not decoding) queries to $\PRC$; by pseudorandomness of $\PRC$ the challenger's responses are indistinguishable from uniform randomness.

This completes the proof.
\end{proof}

\section{CCA security: the public-key setting}
\label{section:cca-security}

This section is dedicated to extending the notion of {\em public-key} adaptively secure PRC to what we call {\em CCA-secure (public-key) PRC}.
The corresponding definition is given in Section~\ref{subsection:cca-definition}.
We then show how to {\em generically} build such a CCA-secure PRC from {\em any} public-key single-bit adaptively secure PRC --- such as the one from Section~\ref{subsection:single-bit} --- in two-steps.
First, in Section~\ref{subsection:rate-boost-pk} we generically show how to boost the information rate of adaptively-secure public-key PRC to become become a constant.
Unfortunately, this suffers from worse --- but still constant --- robustness as compared to the secret-key transformation described in~\Cref{subsection:rate-boost-sk}.
Then, in Section~\ref{subsection:cca-scheme}, we (also generically) show how to upgrade the public-key adaptive security from Section~\ref{sec:def-adaptive} to our notion of CCA security in Section~\ref{subsection:cca-definition}.
This transformation preserves the information and error rates of the corresponding adaptively-secure public-key PRC.
However, we are only able to prove its security in the random oracle model.

\subsection{CCA security definition}
\label{subsection:cca-definition}

\noindent We present a stronger definition of public-key PRC that captures both pseudorandomness and adaptive robustness.
We refer to this security definition as CCA (Chosen Codeword Attack) security, in analogy with the related Chosen Ciphertext Attack security definition of normal encryption schemes.
\par Consider the following security games defined for a pseudorandom code $\PRC=(\PRC.\KeyGen,\PRC.\Encode,\allowbreak \PRC.\Decode)$. The goal of the adversary is to determine whether it is given oracle access to the encryption oracle or to a random codeword oracle, even in the presence of the decoding oracle. Simultaneously, robustness is ensured by modifying the decoding oracle in the ``random case'' to always output the message for codewords close to previously returned codewords. 
%the actual decoder of the PRC, or to an random decoding oracle.

\noindent \underline{$\Gpub_{\adv, \delta, \PRC}(1^\secpar,b)$:}
\begin{enumerate}
    \item The challenger samples $(\pk,\sk) \gets \PRC.\KeyGen(1^\secpar)$. It sends $\pk$ to $\adv$. 
    \item The adversary is allowed to make encoding and decoding queries. Depending on the value of the bit $b$, the challenger responds. 
    \begin{itemize} 
    \item If $b=0$: 
    \begin{itemize}
        \item For each encoding query $m$, the challenger responds with $\PRC.\Encode(\pk, m)$.
        \item For each decoding query $x$, the challenger responds with $\PRC.\Decode(\sk, x)$.
    \end{itemize}
    \item If $b=1$: 
    \begin{itemize}
        \item For each encoding query $m$, the challenger responds with a fresh random string $c$ and sets $\transcript = \transcript \cup \{(m, c)\}$.
        \item For each decoding query $x$, the challenger responds with a random $m$ such that $(m, c) \in \transcript$ and $\wt(c \oplus x) \leq \delta n$. If no such $m$ exists, the challenger responds with $\PRC.\Decode(\sk, x)$.
    \end{itemize}
    \end{itemize} 
    \item The adversary returns a bit $\hat{b}$.
\end{enumerate}

\begin{definition}[CCA security for PRCs] \label{definition:cca-security}
    We say that a secret-key pseudorandom code $\PRC$ satisfies $\delta$-\emph{CCA security} if, for any efficient adversary $\adv$,
    $$\abs{\Pr[\Gpub_{\adv,\delta,\PRC}(1^\secpar,0) = 1] - \Pr[\Gpub_{\adv, \delta,\PRC}(1^\secpar,1) = 1]} \leq \negl.$$
\end{definition}

\begin{remark*}
    Notice that when $\delta=0$, this definition is equivalent to the (simulation-based)  notion of CCA-security {\em with pseudorandom ciphertexts}, where the challenge ciphertext is either the encryption of the challenge message, or random. Normally, the attacker is prohibited from decrypting such challenge ciphertext, which is of course equivalent to modifying the decryption oracle to return the challenge message even in the random case. 
\end{remark*}

\subsection{Boosting the information rate}
\label{subsection:rate-boost-pk}
We now show that one can boost an adaptively robust single-bit public-key PRC to a one with a linear information rate.
This transformation incurs a significant reduction (from $1/4$ to at best $1/32$) in the (albeit still constant) error rate tolerated, which we do not attempt to optimize.

Let $\PRC_1 = (\PRC_1.\KeyGen, \PRC_1.\Encode, \PRC_1.\Decode)$ be a single-bit public-key PRC with codewords of length $\secpar$ that is adaptively robust up to error rate $\beta \in (0, 1/4)$. 
Let $\ECC_1 = (\Enc_1, \Dec_1)$ and $\ECC_2 = (\Enc_2, \Dec_2)$ be error-correcting codes from $\{0,1\}^{\secpar}$ to $\{0,1\}^{\ell j}$ for $\ell \geq \secpar$, and $\{0,1\}^{\secpar \ell}$ to $\{0,1\}^{\secpar \ell j}$, that are robust up to error rate $\alpha \in (0, 1/4)$. 
There are many standard codes which achieve this desired adaptive robustness for any message length, with $j$ equal to a constant.
We also let $\PRG : \{0,1\}^\secpar \to \{0,1\}^{\secpar \ell j}$ be a pseudorandom generator.

\par Then we define a public-key PRC, $\PRC_k = (\KeyGen_k, \Encode_k, \Decode_k)$, which for $k = \secpar \ell$ encodes $k$-bit messages into $2k j$-bit codewords, as follows:

\noindent $\PRC_k[\ECC_1, \ECC_2, \PRC_1, \PRG] = (\KeyGen_k, \Encode_k, \Decode_k)$:
\begin{itemize}
    \item $\KeyGen_k(1^\secpar)$: Output $(\pk,\sk) \gets \PRC_1.\KeyGen(1^\secpar)$.
    
    \item $\Encode_k(\pk, \m)$: Draw a random $r \gets \{0,1\}^{\secpar}$, and let $r' = \Enc_1(r)$.
    For each $i \in [k]$, let $x_i \gets \PRC_1.\Encode(\pk, r'_i)$.
    Output
    \[
    x_1 || \ldots || x_k || \PRG(r) \oplus \Enc_2(\m).
    \]
    
    \item $\Decode_k(\sk, c)$:
    Parse $c$ as $k$ blocks of length $n$, $x_1 || \ldots || x_{k}$, followed by a block $y$ of length $nk$.
    Let $s_i = \PRC_1.\Decode(\sk, x_i)$ for each $i \in [k]$.
    Let $r = \Dec_1(s_1 || \ldots || s_k)$.
    Output 
    \[
    \Dec_2(\PRG(r) \oplus y).
    \]
\end{itemize}

\begin{theorem} \label{theorem:rate-boost-pk}
    If $\ECC_1$ and $\ECC_2$ be error-correcting codes which are (worst-case) robust to substitutions up to rate $\alpha \in (0, 1/4)$, and $\PRC_1$ is a pseudorandom code which is adaptively robust to substitutions up to rate $\beta \in (0,1/4)$. Then $\PRC_k$ above is adaptively robust to substitutions, up to error rate $\alpha \beta/2$.
    Furthermore $\PRC_k$ has a linear information rate.
\end{theorem}
\begin{proof}
    Let $x_1, \ldots, x_k, y$ denote the blocks of a given codeword as defined in $\Decode_k$.
    Recall that $x_1|| \ldots || x_k$ has length $nk$, and $y$ has length $nk$ as well.

    We'll first show that we can recover $r$ from $x_1, \ldots, x_k$.
    Since these blocks make up half of the codeword, the adversary can introduce a rate of $\alpha\beta$ errors here.
    Therefore, at most an $\alpha$ fraction of blocks $x_i$ have an error rate of greater than $\beta$.
    By adaptive robustness of $\PRC_1$, this $1-\alpha$ fraction of blocks all decode correctly to $r'_i$.
    By adaptive robustness of $\ECC_1$, decoding recovers $r$ from the $r'_i$'s.

    We now show that given $r$, one can recover the message from $y$.
    Since $y$ makes up half of the codeword, the adversary can introduce at most an $\alpha \beta < \alpha$ rate of errors to $y$.
    Therefore, adaptive robustness of $\ECC_2$ ensures that $\Dec_2(\PRG(r) \oplus y)$ yields $\m$.

    Finally, observe that $\PRC_k$ encodes $k$-bit messages using codewords of length $2k j$.
\end{proof}

\subsection{Construction and security proof}\label{subsection:cca-scheme}
\noindent Inspired by the Fujisaki-Okamoto (FO) transformation~\cite{FO99} from semantic security to CCA security, we present a transformation from public-key adaptively robust PRC to public-key CCA-secure PRC in the random oracle model.  

\renewcommand{\PRCsharppub}{{\sf PRC}^{\sf CCA}}
\paragraph{Construction.} Let $\PRC$ be a {\em public-key} pseudorandom code for multi-bit messages with adaptive robustness up to $\delta$ errors. Let ${\cal H}_{n,m}=\left\{H:\{0,1\}^n \rightarrow \{0,1\}^m \right\}$ be a hash function family. %This construction is inspired by the Fujisaki-Okamoto transform~\cite{FO99}. 
We define a public-key pseudorandom code $\PRCsharppub$, satisfying $\delta$-CCA security, below.

%\noindent $\PRCsharppub[\calH_{n,m}, \PRC, \delta] = (\KeyGen^{{\sf CCA}}, \Encode^{{\sf CCA}}, \Decode^{{\sf CCA}})$:
\begin{itemize}
    \item $\KeyGen^{{\sf CCA}}(1^\secpar)$: Sample $H \leftarrow {\cal H}_{n,m}$, where polynomials $n(\secparam) \geq \lambda$, $m(\secparam) \geq 2\lambda$ and $(\PRC.\pk,\PRC.\sk) \gets \PRC.\KeyGen(1^\secpar)$. Output $\sk = \PRC.\sk$ and $\pk = (\PRC.\pk, H)$.
    
    \item $\Encode^{{\sf CCA}}(\pk, \m)$: Sample $r \from \{0,1\}^\secpar$ and let $(R_1, R_2) = H(r||\m)$, where $R_1 \geq \lambda$ and $R_2 \geq \lambda$.
     
    Output $\PRC.\Encode(\PRC.\pk, r||\m||R_2; R_1)$.
    
    \item $\Decode^{{\sf CCA}}(\sk, c)$: Compute $r||\m||R_2 = \PRC.\Decode(\PRC.\sk, c)$ and let $(R_1, \tilde{R}_2) = H(r||m)$. If
    \[
        \tilde{R}_2 = R_2 \text{ and }\wt(\PRC.\Encode(\PRC.\pk, r||\m||R_2; R_1) \oplus c) \le \delta n,
    \]
    output $\m$; otherwise, output $\bot$.
\end{itemize}

\noindent We show that the above construction satisfies CCA security. 

\begin{theorem}
\label{thm:prcsharppub:cca}
$\PRCsharppub$ satisfies $\delta$-CCA security. 
\end{theorem}
\begin{proof}
Suppose $\adv$ is a probabilistic polynomial time adversary participating in the experiment $\Gpub_{\adv, \delta, \PRCsharppub}(1^\secpar,b)$. We establish some notation first.
\begin{itemize}
    \item Let $\ell$ be the number of encoding and decoding queries made by $\adv$; we can assume that the number of encoding and decoding queries is the same without loss of generality. Recall that $\adv$ is allowed to intersperse the encoding and decoding queries arbitrarily. 
    \item Denote $(m_1,\ldots,m_{\ell})$ to be the sequence of (adaptive) encoding queries made by $\adv$. 
    \item For the $i^{th}$ encoding query $m_i$, the challenger does the following: it picks $r_i \from \{0,1\}^\secpar$. It then queries $(m_i,r_i)$ to the random oracle (RO) to obtain $H(m_i||r_i)=((R_1)_{i}||(R_2)_{i})$. Then it sets  $c_i=\Encode^{{\sf CCA}}(\pk,(r_i||m_i||(R_2)_{i});(R_1)_{i})$. It responds to $\adv$ with $c_i$. Thus, $(c_1,\ldots,c_{\ell})$ is the list of responses to the encoding queries. 
       \item Denote $Q_{Ch}^{{\sf RO}} = \left\{ (m_1,r_1),\ldots,(m_{\ell},r_{\ell})\right\}$ to be the set of RO queries made by the challenger during the encoding queries.
    \item Denote $(c'_1,\ldots,c'_{\ell})$ to be the sequence of (adaptive) decoding queries.
    \item For the $i^{th}$ decoding query $c'_i$, the challenger does the following: it returns the output of $\Decode^{{\sf CCA}}(\sk,c'_i)$ to $\adv$. In more detail, $\Decode^{{\sf CCA}}$ first computes $(m'_i,r'_i,(R_2)'_i) \leftarrow \PRC.\Decode(\sk,c'_i)$ and then returns $m'_i$ to $\adv$ if and only if $(R_2)'_i=(R_2)''_i$, where $H(m_i||r_i)=((R_1)''_i||(R_2)''_i)$. Otherwise, it returns $\bot$. Denote $Q_{Ch}^{{\sf Resp}}$ to be the set of responses to the adversarial decoding queries. That is, $Q_{Ch}^{{\sf Resp}}=\{m'_1,\ldots \allowbreak,m'_q \}$. Let $Q_{Ch}^{{\sf Dec}}$ be the set of decoded outputs of the adversarial decoding queries. That is, $Q_{Ch}^{{\sf Dec}}=\{(m'_1,r'_1,(R_2)'_1),\ldots \allowbreak,(m'_q,r'_q,(R_2)'_q) \}$. Note that $Q_{Ch}^{{\sf Resp}}$ can be obtained by considering the first components of all the elements in $Q_{Ch}^{{\sf Dec}}$.

    \item Denote $Q_{\adv}^{{\sf RO}} = \left\{(m''_1,r''_1),\ldots,(m''_Q,r''_Q) \right\}$ to be the set of RO queries made by the adversary.
    
\end{itemize}

\newcommand{\hybrid}{\mathsf{Hyb}}
\newcommand{\baddecodingquery}{\mathbf{BadDecQuery}}

\noindent We describe the hybrids below. 

\noindent \underline{{\bf Hybrid} $\hybrid_1$}: $\Gpub_{\adv, \delta, \PRCsharppub}(1^\secpar,0)$.

\noindent Before describing $\hybrid_2$, we define an event below. 

\begin{quote}
\noindent \underline{${\sf NeverQueried}$}: There exists $(m'_i,r'_i,(R_2)'_i)$ such that $(m'_i,r'_i)$ has never been queried by $\adv$ to the RO until the point where the adversary makes the $i^{th}$ decoding query $c'_i$. Recall that $(m'_i,r'_i,(R_2)'_i) \in Q_{Ch}^{{\sf Dec}}$. 
\end{quote} 

%\noindent ${\sf CloseQueries}$: There exists $(m_1^*,r_1^*) \in Q_{\adv}^{{\sf RO}} \cup Q_{Ch}^{{\sf RO}}$ and $(m_2^*,r_2^*) \in Q_{\adv}^{{\sf RO}} \cup Q_{Ch}^{{\sf RO}}$ such that (a) $\wt(c_1^* \oplus c_2^*) \leq 2 \delta n$ and (b) $m_1^* \neq m_2^*$, where for every $i \in \{1,2\}$, $c_i^* = \Encode(\pk,(m_i^*,r_i^*,(R_2)_i^*);(R_1)_i^*)$ and $H(m_i^*,r_i^*)=((R_1)_i^*||(R_2)_i^*)$.

\noindent \underline{{\bf Hybrid} $\hybrid_2$}: This hybrid is identical to the previous hybrid except that if ${\sf NeverQueried}$ happens at any point, immediately abort. 

\begin{lemma}
\label{lem:hybrid1}
$\prob{{\sf NeverQueried}} =\negl$.
\end{lemma}
\begin{proof}
We first show that $\prob{{\sf NeverQueried}}$ is negligible. Note that the probability that we can predict the output of $H$ on $x$ where $x$ has never been queried before is at most $\frac{1}{2^{\lambda}}$. Since $\adv$ can submit $\ell$ decoding queries, the probability that $\adv$ can successfully predict the output of $H$ is at most $\frac{\ell}{2^{\lambda}}$.
\end{proof}

\begin{lemma}
\label{lem:hybrid1:hybrid2}
Hybrids $\hybrid_1$ and $\hybrid_2$ can be distinguished by $\adv$ with advantage at most $\negl$. 
\end{lemma}
\begin{proof}
Conditioned on ${\sf NeverQueried}$ never happening, hybrids $\hybrid_1$ and $\hybrid_2$ are identical. This observation along with~\Cref{lem:hybrid1} completes the proof. 
\end{proof}

\noindent \underline{{\bf Hybrid} $\hybrid_3$}: We first define an alternate decoding procedure, denoted by ${\sf AltDecode}$. 

\begin{quote}
\underline{${\sf AltDecode}$}: on input a codeword $c$ and a set $\{(m''_1,r''_1),\ldots,(m''_q,r''_q),m_{i_1},c_{i_1},\ldots,m_{i_{t}},c_{i_t}\}$, 
\begin{itemize} 
\item It first computes the codewords $\{{c}''_1,\ldots,{c}''_q\}$, where ${c}''_i = \PRC.\Encode(\PRC.\pk,r''_i||m''_i||(R_2)''_i;(R_1)''_i)$ and $H(r''_i||m''_i) = ((R_1)''_i,(R_2)''_i)$. 
\item It then checks\footnote{It could be the case that $c$ is close to more than one codeword in the set $\{c''_{1},\ldots,c''_q\} \cup \{c_{i_1},\ldots,c_{i_t}\}$. In this case, ${\sf AltDecode}$ picks one of the matched codewords at random.} if $\wt(c \oplus c''_j) \leq \delta n$ for some $j \in [q]$. It also checks if $\wt(c \oplus c_{i_k}) \leq \delta n$, for some $k \in [t]$. If such a $j$ (or $k$) exists then it outputs $m''_{j}$ (or $m_{i_k}$). Otherwise, it outputs $\bot$.
\end{itemize} 
\end{quote} 
\par This hybrid is identical to the previous hybrid except that the above alternate decoding procedure ${\sf AltDecode}$ is used to answer decoding queries. In more detail, whenever $\adv$ submits the $i^{th}$ decoding query $c'_i$, run ${\sf Altdecode}(c'_i,\allowbreak \{(m''_1,r''_1),\allowbreak \ldots,\allowbreak (m''_q,r''_q),m_{i_1},c_{i_1},\ldots,m_{i_t},c_{i_t}\})$, where $\{(m''_1,r''_1),\allowbreak \ldots \allowbreak,(m''_q,r''_q)\}$ is the set of recorded RO queries made so far by $\adv$ and $c_{i_1},\ldots,c_{i_t}$ is the set of codewords generated for the messages $m_{i_1},\ldots,m_{i_t}$ during the encoding queries made so far by $\adv$. Denote the result to be $m_{{\sf alt}}^{(i)}$. Also, run the real decoding procedure $\Decode(\sk,c)$ to obtain $m_{{\sf real}}^{(i)}$. If $m_{{\sf alt}}^{(i)} \neq m_{{\sf real}}^{(i)}$, abort the experiment. Otherwise, return $m_{{\sf alt}}^{(i)}$ to $\adv$. 

%Replace the decoding oracle with an alternate decoding procedure described as follows: instead of using the decoding key, the challenger now does the following. It first computes the codewords $\{{c}''_1,\ldots,{c}''_q\}$, where ${c}''_i = \PRC.\Encode(\PRC.\pk,r''_i||m''_i||(R_2)''_i;(R_1)''_i)$ and $H(r''_i||m''_i) = ((R_1)''_i,(R_2)''_i)$; note that this process can be carried out by observing the RO queries made by the adversary. On input a decoding query $c''_i$, the challenger checks if $\wt(c'_i \oplus c''_j) \leq \delta n$ for some $j$. If such a $j$ exists then it returns $m''_j$. Else, it returns $\bot$. The rest of the hybrid is the same as $\hybrid_2$. \\

\begin{lemma}
\label{lem:hybrid2:hybrid3}
Assuming the adaptive robustness property of $\PRC$, hybrids $\hybrid_2$ and $\hybrid_3$ can be distinguished by $\adv$ with advantage at most $\negl$. 
\end{lemma}
\begin{proof}
Conditioned on $m_{{\sf alt}}^{(i)} = m_{{\sf real}}^{(i)}$, for every $i \in [\ell]$ -- i.e. the outputs of the alternate and the real decoding procedures are the same -- the hybrids $\hybrid_2$ and $\hybrid_3$ are identical. Consider the following event. 

\begin{quote}
\noindent \underline{${\sf UnEqual}$}: there exists $i^* \in [\ell]$ such that $m_{{\sf alt}}^{(i^*)} \neq m_{{\sf real}}^{(i^*)}$.
\end{quote}

\noindent We claim that $\prob{{\sf UnEqual}}$ is $\negl$. Suppose $\prob{{\sf UnEqual}}$ is non-negligible, we violate the adaptive robustness property of $\PRC$.

\par We design the following reduction that violates the adaptive robustness property of $\PRC$: 
\begin{itemize}
    \item It samples $\widehat{i} \xleftarrow{\$}  [\ell]$. 
    \item It employs lazy sampling to simulate the RO queries made by $\adv$.
    \item It duly forwards the public key it receives from the external challenger to $\adv$. 
   \item Encoding queries: for the $i^{th}$ query $m_i$ from $\adv$, the reduction samples $r_i$ uniformly at random and submits $(m_i,r_i)$ to the external challenger. It receives a codeword $c_i$ from the external challenger which it duly forwards to $\adv$.   
    \item Decoding queries: for the $i^{th}$ decoding query $c'_{i}$, it first checks if $i < \widehat{i}$. If so, it employs the alternate decoding procedure ${\sf AltDecode}$ to respond to decoding queries. For the $(\widehat{i})^{th}$ decoding query, it checks if $m_{{\sf alt}}^{(\widehat{i})} \neq m_{{\sf real}}^{(\widehat{i})}$, where $m_{{\sf alt}}^{(\widehat{i})}$ and $m_{{\sf real}}^{(\widehat{i})}$ are as defined in the description of $\hybrid_3$. If the check did not pass, it aborts the experiment. Otherwise, it sets $((m,r,R_2),R_1)$ as follows: (a) $c \leftarrow \PRC.\Encode(\PRC.\pk,(m,r,R_2);R_1)$, (b) $\wt(c,c'_{\widehat{i}}) \leq \delta n$ and, (c) $m=m_{{\sf alt}}^{(\widehat{i})}$. It sends $((m,r),\ (R_2),\ c'_{\widehat{i}})$ to the external challenger and ends the adaptive robustness experiment. In other words, after the $(\widehat{i})^{th}$ decoding query, the execution of the reduction ends. 
    \item RO queries: The reduction answers consistently according to the table generated during the lazy sampling procedure.
\end{itemize}
\noindent With probability $\frac{1}{\ell} \prob{{\sf UnEqual}}$, which is non-negligible by our assumption, the following events holds: $(\widehat{i})^{th}$ decoding query is the first decoding query where the event ${\sf UnEqual}$ holds. 
% \begin{itemize}
%     \item $(\widehat{i})^{th}$ decoding query is the first decoding query where the event ${\sf UnEqual}$ holds. 
%     \item $\adv$ never queries the RO on $(m_i,r_i)$, for any $i \in [\ell]$, before the $(\widehat{i})^{th}$ decoding query. 
% \end{itemize}

This shows that the reduction violates the adaptive robustness property of $\PRC$ with non-negligible probability, which is a contradiction.  
\end{proof}
%\noindent \underline{$\hybrid_{4}$}: Modify the game such that it aborts if the adversary and challenger ever query the random oracle on the same $r$.
%That is, it aborts if the adversary queries the random oracle on $(* || r_i)$ for any $r_i$ previously queried by the challenger.
%It aborts if the challenger ever queries the random oracle on $(* || r'_i)$ for any $r'_i$ previous queried by the adversary.

\noindent \underline{{\bf Hybrid} $\hybrid_4$}: Sample $(c_1,\ldots,c_{\ell})$ instead from the uniform distribution. 

\begin{lemma}
Assuming the pseudorandomness of $\PRC$, the hybrids $\hybrid_3$ and $\hybrid_4$ can be distinguished by $\adv$ with advantage at most $\negl$. 
\end{lemma}
\begin{proof}
%Suppose the hybrids $\hybrid_3$ and $\hybrid_4$ can be distinguished by $\adv$ with probability $\varepsilon$. We prove that $\varepsilon$ is negligible by contradiction. Without loss of generality, we assume that 
%\par We describe the following reduction. 

We define the following event. 

\begin{quote}
\noindent \underline{${\sf BadDecodingQuery}$}: If there exists $i$ such that $(m_i,r_i) \in Q_{\adv}^{{\sf RO}}$. 
\end{quote} 

\noindent We first claim that the probability that ${\sf BadDecodingQuery}$ happens in $\hybrid_3$ is $p=\negl$. 
\par Suppose not; that is, let $p$ be non-negligible.  We show that we can violate the pseudorandomness of $\PRC$. The reduction, that violates the pseudorandomness of $\PRC$, does the following: 
\begin{itemize}
    \item It employs lazy sampling to simulate the RO queries made by $\adv$.
    \item It duly forwards the public key it receives from the external challenger to $\adv$. 
    \item Encoding queries: for the $i^{th}$ query $m_i$ from $\adv$, the reduction samples $r_i$ uniformly at random and submits $(m_i,r_i)$ to the external challenger. It receives a codeword $c_i$ from the external challenger which it duly forwards to $\adv$.   
    \item Decoding queries: it employs the alternate decoding procedure ${\sf AltDecode}$ described in $\hybrid_3$ to respond to decoding queries. 
    \item RO queries: if the adversary ever queries $(m_i,r_i)$, the reduction stops the execution and outputs 1. Otherwise, it answers consistently according to the table generated during the lazy sampling procedure.   
\end{itemize}
\noindent If $\adv$ finishes its execution successfully then the reduction outputs 0. 
\par Note that the probability that the reduction outputs 1 in the case when the reduction receives pseudorandom codewords from the external challenger (as generated in $\hybrid_3$) is non-negligible. This is due to our hypothesis that the probability that ${\sf BadDecodingQuery}$ happens in $\hybrid_3$ is $p$, which is assumed to be non-negligible. On the other hand, the probability that the reduction outputs 1 in the case when the reduction receives uniformly random codewords from the external challenger is $\negl$. This is due to the fact that $\{r_1,\ldots,r_{\ell}\}$ is information-theoretically hidden from $\adv$. Hence, the probability that the adversary queries on $(m_i,r_i)$, for any $i$, is at most $\frac{\ell}{2^{\secpar}}$. Thus, the reduction successfully violates the pseudorandomness of $\PRC$, a contradiction. From this we can conclude that $p$ has to be $\negl$. 

\noindent We are now ready to prove the lemma. Suppose the hybrids $\hybrid_3$ and $\hybrid_4$ can be distinguished by $\adv$ with probability $\varepsilon$. We prove by contradiction that $\varepsilon$ is non-negligible. Consider the following reduction: 
\begin{itemize}
    \item It employs lazy sampling to simulate the RO queries made by $\adv$.
    \item It duly forwards the public key it receives from the external challenger to $\adv$. 
    \item Encoding queries: for the $i^{th}$ query $m_i$ from $\adv$, the reduction samples $r_i$ uniformly at random and submits $(m_i,r_i)$ to the external challenger. It receives a codeword $c_i$ from the external challenger which it duly forwards to $\adv$.   
    \item Decoding queries: it employs the alternate decoding procedure described in $\hybrid_3$ to respond to decoding queries. 
    \item RO queries: if the adversary ever queries $(m_i,r_i)$, the reduction aborts. Otherwise, it answers consistently according to the table generated during the lazy sampling procedure.   
\end{itemize}
\noindent After the execution of $\adv$, the output of the reduction is set to be the output of $\adv$.  
\par Since the probability that ${\sf BadDecodingQuery}$ happens is negligible in both the hybrids $\hybrid_3$ and $\hybrid_4$, the probability that the reduction aborts (in both the cases) is negligible. Conditioned on the reduction never aborting, the simulation of $\adv$ by the reduction in the case when $\{c_1,\ldots,c_{\ell}\}$ are pseudorandom (resp., uniform) is identical to $\hybrid_3$ (resp., $\hybrid_4$). Thus, the probability that the reduction violates the pseudorandomness of $\PRC$ is negligibly close to $\varepsilon$, which is non-negligible. This contradicts the pseudorandomness of $\PRC$. Thus, $\varepsilon$ is negligible, which completes the proof. 

\end{proof}

\noindent \underline{{\bf Hybrid} $\hybrid_5$}: This is identical to $\Gpub_{\adv, \delta, \PRCsharppub}(1^\secpar,1)$ except that if  ${\sf NeverQueried}$ ever happens, abort. \\

\begin{lemma}
Hybrids $\hybrid_4$ and $\hybrid_5$ can be distinguished by $\adv$ with advantage at most $\negl$.     
\end{lemma}
\begin{proof}
This proof is similar to the proof of~\Cref{lem:hybrid2:hybrid3}. 
\end{proof}

\noindent \underline{{\bf Hybrid} $\hybrid_6$}: $\Gpub_{\adv, \delta, \PRCsharppub}(1^\secpar,1)$. \\

\begin{lemma}
Hybrids $\hybrid_5$ and $\hybrid_6$ can be distinguished by $\adv$ with advantage at most $\negl$.     
\end{lemma}
\begin{proof}
This proof is similar to the proof of~\Cref{lem:hybrid1:hybrid2}. 
\end{proof}

\noindent This completes the proof of~\Cref{thm:prcsharppub:cca}. 
\end{proof}

\paragraph{Acknowledgments.}
Omar Alrabiah was supported in part by a Saudi Arabian Cultural Mission (SACM) Scholarship, NSF CCF-2210823 and V.\ Guruswami's Simons Investigator Award.
Prabhanjan Ananth was supported by NSF CNS-2329938 and NSF Career-2341004.
Miranda Christ was supported by a Google CyberNYC grant, an Amazon Research Award, and NSF grants CCF-2312242, CCF-2107187, and CCF-2212233.
Yevgeniy Dodis was partially supported by NSF grant CNS-2055578, and a gift from Google.
Sam Gunn was partially supported by a Google PhD fellowship.
Any opinions, findings and conclusions or recommendations expressed in
this material are those of the authors and do not necessarily reflect
the views of the United States Government, Amazon, Google, or any other
supporting organization.

We thank Shivam Nadimpalli for helpful conversations in the early stages of this work.

\newpage
\bibliographystyle{alpha}
\bibliography{00-references.bib}

\newpage
\appendix

\section{Related work} \label{section:related-work}
We are motivated by a line of work on embedding hidden patterns into the outputs of generative models.
To our knowledge, the first of these works was \cite{ZDR19} in the context of steganography using text.
Later, \cite{KJGR21} observed that the scheme of \cite{ZDR19} could be made secure by applying pseudorandom encryption (which can be built from one-way functions).
However, both of these works require full knowledge of the prompt and that there are no modifications to the text, rendering them inapplicable to watermarking.

Due to the rise of generative AI, there has been a renewed interest in watermarking.
Starting with \cite{Aar22,KGW+23}, which focus on language models, there has emerged a line of work that embeds watermarks by modifying the randomness used in these generative algorithms.
In \cite{CGZ24}, they showed that it is possible to modify the randomness of a language model such that the distribution of the watermarked model is \emph{computationally indistinguishable} from that of the original model.
This indistinguishability property, called \emph{undetectability}, implies that the watermark does not degrade the quality of the model under any efficiently computable metric.
Undetectability is identical to steganographic secrecy, and an undetectable multi-bit watermark implies stateless, robust steganography.

While the watermark of \cite{CGZ24} achieves a very strong quality guarantee, it has only very weak \emph{robustness} --- that is, a small number of edits can remove the watermark from the text.
This is because it relies on using a security-parameter-length substring of the text to seed a pseudorandom function used to embed a watermark in the remainder of the response.
If any word in this substring is changed, the seed cannot be recovered at detection time and the watermark disappears.
Therefore the \cite{CGZ24} scheme is only robust against a very constrained adversary that may crop a response to a sufficiently long contiguous substring, but not replace words.

In \cite{FGJ+23}, they introduce a notion of publicly-detectable watermarks.
While it is of course possible to publish the secret watermarking key for any scheme to make it ``publicly detectable,'' the feature that distinguishes \cite{FGJ+23} is that the watermark is \emph{unforgeable} even by an adversary who knows the detection key.
However, their notion of public detectability directly contradicts their definition of robustness --- an issue that other works overcome by introducing two separate watermark detection keys \cite{CG24,ZZXR24}.
Furthermore, this scheme (even the version that uses an error-correcting code) has the same level of very-weak robustness as that of \cite{CGZ24}, and is only undetectable under the assumption that every response from the model has high entropy per fixed-length ``chunk'' of tokens.

Other text watermarks such as \cite{ZALW24,KTHL24} achieve stronger robustness, to even a constant rate of edits, but at the cost of significant degradation in quality.
This state of affairs suggested to many a fundamental tradeoff between quality and robustness of text watermarks.

Pseudorandom codes (PRCs) were introduced to overcome this trade-off, allowing for undetectable watermarks with robustness to even a constant rate of substitutions and random deletions \cite{CG24}.
In the same paper, they also demonstrate that PRCs immediately allow for unforgeable publicly-detectable watermarks, overcoming the aforementioned issues with \cite{FGJ+23}.
Since \cite{CG24}, a few works have used pseudorandom codes and their watermarking framework to construct language model watermarks \cite{GM24,GG24}, as well as an image model watermark \cite{GZS24}.

As stronger robustness of a PRC directly translates to stronger watermark robustness, some works have focused on constructing PRCs with stronger robustness guarantees.
In \cite{CHS24}, they consider a notion of adaptive robustness of a watermark, where the errors are allowed to depend on previously seen watermarked responses.
They prove that the scheme of \cite{CGZ24} is adaptively robust, but as discussed above, that scheme is only very-weakly robust.
They conjecture that the PRC of \cite{CG24} is adaptively robust to even a constant rate of substitutions.

Subsequently, Golowich and Moitra \cite{GM24} showed how to construct binary PRCs from an alternative assumption --- namely, the existence of a family of $(\log n)$-juntas that is hard to learn.
They then presented a powerful and generic transformation that converts any binary-alphabet PRC into a polynomial-sized-alphabet PRC with robustness to non-adaptive edits.

The most recent work on constructing PRCs is that of Ghentiyala and Guruswami \cite{GG24}.
They first show that, assuming just the existence of one-way functions, there exist PRCs with robustness to any non-adaptive channel that introduces any $o(1)$ rate of errors.
They then construct PRCs under alternative notions of pseudorandomness, where the distinguisher is space-bounded or is allowed $o(1)$ distinguishing advantage.
They also present an interesting variant of the codes from \cite{CG24} that permits pseudorandomness to be based on a wider range of assumptions.

Although work on watermarks offering provable guarantees has focused largely on language models, the framework of \cite{CG24} is applicable to generative AI more broadly.
For instance, \cite{GZS24} used PRCs to build a practical image watermark for diffusion models that is quality-preserving.
Therefore, improving PRCs is an avenue for improving watermarks \emph{in general}.

\end{document}